\documentclass[preprint,12pt,3p]{elsarticle}

\usepackage{graphicx}
\usepackage{multirow}
\usepackage{xcolor}
\usepackage{amsmath,amssymb,amsthm,bbm}
\usepackage[mathscr]{euscript} 
\usepackage{accents} 
\usepackage[noend]{algorithm2e}  
\usepackage{float}
\newcommand{\indic}[1]{\ensuremath{ \mathbbm{1}_{#1} } }
\newcommand{\mytilde}[1]{\ensuremath\accentset{\sim}{ #1 }}
\newtheorem{Theorem}{Theorem} 
\makeatletter
\newsavebox\myboxA
\newsavebox\myboxB
\newlength\mylenA
\newcommand*\xoverline[2][0.75]{%
    \sbox{\myboxA}{$\m@th#2$}%
    \setbox\myboxB\null
    \ht\myboxB=\ht\myboxA%
    \dp\myboxB=\dp\myboxA%
    \wd\myboxB=#1\wd\myboxA
    \sbox\myboxB{$\m@th\overline{\copy\myboxB}$}
    \setlength\mylenA{\the\wd\myboxA}
    \addtolength\mylenA{-\the\wd\myboxB}%
    \ifdim\wd\myboxB<\wd\myboxA%
       \rlap{\hskip 0.5\mylenA\usebox\myboxB}{\usebox\myboxA}%
    \else
        \hskip -0.5\mylenA\rlap{\usebox\myboxA}{\hskip 0.5\mylenA\usebox\myboxB}%
    \fi}
\makeatother

\journal{Chaos: An Interdisciplinary Journal of Nonlinear Science}

\begin{document}

\begin{frontmatter}

\title{Application of the interacting particle system method  to  piecewise deterministic Markov processes  used in reliability} 

\author[edf]{H. CHRAIBI }
\author[edf]{A. DUTFOY }

\author[edf,lpma]{\emph{T. GALTIER}\corref{cor1} }
\cortext[cor1]{Corresponding author }  
\ead{tgaltier@gmail.com} 

\author[X]{J. GARNIER }

\address[edf]{EDF R\&D - D\'epartement PERICLES,
7 Boulevard Gaspard Monge,  91120 Palaiseau, France}
\address[lpma]{ Universit\'e Paris-Diderot -  Laboratoire~de~Probabilit\'es~Statistique~et~Mod\'elisation,
 75205 Paris Cedex 13, France } 
\address[X]{ Ecole Polytechnique -  Centre de Math\'ematiques Appliqu\'ees, 91128 Palaiseau Cedex, France}

\begin{abstract}
Variance reduction methods are often needed for the reliability assessment of complex industrial systems, we focus on one variance reduction method in a given context, that is the interacting particle system method (IPS) used on piecewise deterministic Markov processes (PDMP) for reliability assessment . The PDMPs are a very large class of processes which benefit from high modeling capacities, they can model almost any Markovian phenomenon that does not include diffusion. In reliability assessment, the PDMPs modeling industrial systems  generally involve low jump rates and  jump kernels favoring one safe arrival, we call such model a "concentrated PDMP".

Used on such concentrated PDMPs, the IPS is inefficient and does not always provide a variance reduction. Indeed, the efficiency of the IPS method relies on simulating many different trajectories during its propagation steps,  but unfortunately concentrated PDMPs are  likely to  generate the same deterministic trajectories over and over.  We propose an adaptation of the IPS method called IPS+M that reduces this phenomenon. 
The IPS+M consists in modifying the propagation steps of the IPS, by conditioning  the propagation to avoid generating the same trajectories multiple times. We prove that, compared to the IPS, the IPS+M method always provides an estimator with a lower variance.  We also carry out a quick simulation study on a two-components system that confirms this result.
\end{abstract}


\end{frontmatter}
\newpage


\section{Introduction}

For both safety and regulation issues, the reliability of industrial systems has to be assessed. The considered systems (dams or nuclear power plants, for instance) are complex dynamic hybrid systems, so only simulation  methods can be reasonably considered to assess their reliability.
The failure of such a dynamic hybrid system generally corresponds to a physical variable of the system (temperature, pressure, water level) entering a critical region. The simulation of such a system requires to accurately model the trajectory of the physical variables. The evolution of these physical variables are generally determined by simple differential equations derived from the laws of physics. As the physical context giving the differential equations generally depends on the statuses of the multiple components of the systems (on, off, or failed),   the differential equations can change whenever there is a change in the components statuses. To encounter for this hybrid interplay between the discrete process of components' statuses and the continuous evolution of the physical variables, we model the evolution of the state of a system by a piecewise deterministic Markov process (PDMP)  \cite{davis1993markov,zhang2015numerical,de2015numerical,cloez2017probabilistic}. PDMPs   are meant to represent a large class of Markovian processes that do not include diffusion, as such, they benefit from high modeling capacity: they can model many industrial systems. For instance, EDF has recently developed the {PyCATSHOO} toolbox \cite{Pycatshoo2}, which allows the modeling of dynamic hybrid systems, the main option within this toolbox is to evaluate the dependability criteria of the studied systems by  Monte Carlo simulations. \\
  
  As the industrial systems are highly reliable, their failure is a rare event, and the naive Monte-Carlo  simulation method (MC) is computationally too intensive in this context. The objective of our work is to set up new algorithms to accelerate the reliability assessment of such industrial systems.  To do so we want to use   faster methods, such as variance reduction methods.
 A variance reduction method is a method that yields an unbiased estimator with a smaller variance than the Monte-Carlo estimator. The estimation being more accurate, we need less simulation runs to reach the desired accuracy, thus we save computational time. The variance reduction is generally achieved by altering the simulation process, and then   correcting the  bias induced by this modification of simulation process by appropriately weighting each simulation.  Several acceleration methods for variance reduction can be proposed, such as importance sampling methods or particle filter methods (also called subset simulations). In this article we focus on one example of particle filter method: the interacting particle system (IPS)  method \cite{del2005genealogical}.\\

 Unlike in importance sampling methods, in the IPS method we keep simulating according to the original system. Another difference is that we do not simulate directly the trajectories on the entire observation duration,  the observation duration is subdivided into small intervals of time and we simulate the trajectories sequentially one interval of time after the other one.
 These sequential simulations of the trajectories consist alternating between an exploration step and a selection step. During the exploration step we simulate the trajectories on a small time interval, therefore exploring the most probable trajectories on a short horizon of time. Then we apply a selection step on these trajectories  replicating the trajectories which seem "close" to failure  and  giving up  the less "promising" trajectories. At the next exploration step, only   replicated   trajectories are continued, before the next selection, and the next exploration and so on... until each selected trajectory reaches the end of the observation time window. This way the effort of simulation is concentrated on selected trajectories, which have higher chance of becoming a failing trajectory before the end of the observation time window. In the end we get more failing trajectories to fuel our estimation, and, if the selection was well done, the IPS method yields an unbiased estimator with a smaller variance than the MC estimator. \\
 
   When we try to apply the IPS method in order to estimate the failure probability of such  reliable  complex hybrid systems, the IPS methods  turn out to be inefficient.  Indeed the estimation provided by the IPS method often has, in this case,  a higher variance than the one of the Monte-Carlo estimator.  IPS does not perform well in this context because the application case (a reliable  complex hybrid system) makes it hard to conduct the exploration steps of the IPS method  efficiently. Indeed, an industrial system is often modeled  by what we call a "concentrated PDMP", which  is a PDMP with  low jump rates  and   concentrated  jump kernels on boundaries.    The typical jump rate  is low because it is the sum of the failure rates of the components in working condition and of the repair rates of the failed components which are all very low. These failure and repair rates are very low  because the components are reliable and their repairs are slow. The typical jump kernel is concentrated because  most of the probability mass of a jump is concentrated on one (safe) output. Indeed the jump kernel on boundaries model the automatic control mechanisms within the system. During such  a control mechanism there is a small probability that some component(s) fail on demand but the most likely output is that the system jumps on the safe state aimed by the control mechanism, consequently the  probability mass of the jump kernel is concentrated on this output.  
   Due to  these characteristics of the model there is a high probability that no component failure or repair  occurs during the short exploration time, and with  a PDMP it means that all the trajectories are likely to  follow the same deterministic paths. So when we explore the trajectory space  most simulated trajectories end up being the same one, hence limiting our exploration of the trajectory space. To avoid this pitfall, a particular filter was proposed in \cite{whiteley2011monte} that enhances the occurrence of random jumps (failure or repairs) or modifies  the occurrence time of the last jump.  However the proposed method is limited to a  different case of PDMP. This class of PDMP does not include   concentrated PDMP, as it contains only PDMP without boundaries which allows continuous jumps kernels,  i.e. it does not allow to model automatic control mechanisms in components.  Moreover, an a priori bound on the number of jumps in a time interval is required, but in our case we do not have such information.
   
   In order to adapt the IPS to concentrated PDMPS, we propose instead to use an approach based on the memorization method developed in \cite{labeau1996probabilistic}. The idea is to start the exploration by finding the most likely trajectories continuing each batch of replicated trajectories. Then, we condition the rest of the exploration to avoid these trajectories. As a result, the simulated trajectories have much more chances to differ which improves the quality of the exploration and reduce the variance of the estimator. 
   To correct the bias induced by this modification of the simulation process, we have to modify the weight of each trajectory.
   We call our adaptation of the IPS to PDMP the IPS+M   for sequential Monte-Carlo sampler with memorization method. \\

The rest of the paper paper is organized as follows: Section \ref{sec:model} is dedicated to the presentation of our model of the system, Section \ref{sec:IPS} presents the IPS method and introduces the optimal potential functions for this method, 
Section \ref{sec:IPS+} presents the IPS+M method which  adapts the IPS algorithm for PDMPs with low jump rates, Section \ref{sec:Memo} explains how to force the differentiation of the trajectories using the memorization method, and finally in Section \ref{sec:Res} we illustrate the better efficiency of the IPS+M method on a toy example.

\section{A model of the system based on a PDMP with discrete jump kernel} 
\label{sec:model}

\subsection{The model}
We denote by $Z_t$ the state of the system at time $t$. $Z_t$ is the combination of the physical variables of the system, noted $X_t$, and of the statuses of all the components within the system, noted $M_t$: $Z_t=(X_t,M_t)$.  We consider that $X_t\in\mathbb R^d$, and that $M_t\in\mathbb M=\{On,Off,F\}^{N_c}$ where $F$ corresponds to a failed status, and $N_c$ is the number of components in the system. The value of $M_t$ is sometimes referred as the  mode of the system. Here we consider only three categories for the status of a component : $On,Off,F$, but it is possible to include more categories as long as the set of the possible modes $\mathbb M$ stays countable.

The process $Z_t$ is piecewise continuous, and each discontinuity  is called a jump.  Between two jumps there is no change in the components' statuses, and the dynamics of the physical variables  can be expressed thanks to an ordinary differential equation derived from the law of physics: $$\frac{d\,X_t}{d t}= F_{M_t} (X_t).$$
We note $\phi_{(x,m)}(t)$ the solution of this equation when $X_0=x$ and $M_0=m$. Then for any time $s>0,\,$ if $T$ is the time separating $s$ from  the next jump time, we have  $$\forall t\in[0,T) ,\quad Z_{s+t} =\big( X_{s+t},M_s\big)=  \big(\phi_{(X_s,M_s)} (t),M_s\big).$$ Similarly, a flow function on the states can be defined. If $ z=(x,m) $, then we define $ \Phi_{z}(t)=\big(\phi_{(x,m)} (t),m\big)$, and so \begin{equation}
  \forall t\in[0,T) ,\  Z_{t+s} =\Phi_{Z_s}(t) .\label{eq:flow}
\end{equation}

As the physical variables are often continuous, the jumps are essentially used to model changes in the statuses of the components. These jumps can occur for two reasons.\\ Firstly, a  jump can correspond to an automatic control mechanism (See Figure \ref{sautdet}). In a given mode $m$, such mechanism is typically triggered when the physical variables cross some threshold. For each mode $m$, we define an open and connected set  $\Omega_m$, so that these thresholds  determine its boundary $\partial \Omega_m$. So when  $M_t=m$ the value of the physical variables $X_t$ is restricted to the set $\Omega_m\subset \mathbb R^d$. Letting $E_m=\{(x,m), x\in\Omega_m\}$ be the set of the possible states with mode $m$, in terms of state a jump  associated to a control mechanism is triggered whenever the state $Z_{s+t}$ hits the boundary of $E_m$. The set of possible states is therefore defined by: $$E\,=\underset{m\in\mathbb{M}}{\bigcup} E_m.$$ \begin{figure}[t]\centering
\includegraphics[width= 0.8\linewidth]{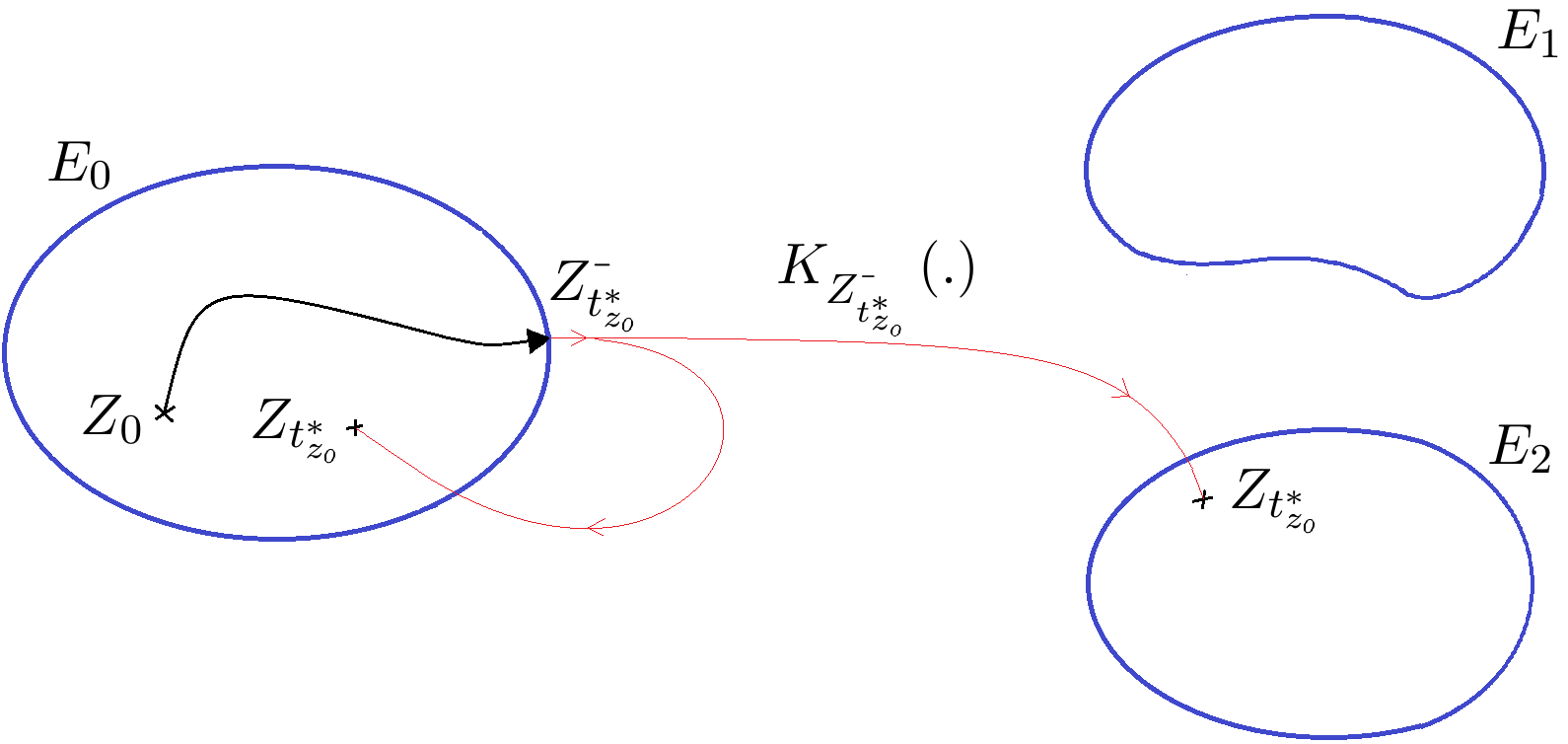}
\caption{A jump at boundary.\label{sautdet}}
\end{figure}  \\Secondly, jumps can correspond to a spontaneous failure or repair. In such case the jump occurs before $Z_{s+t}$ hits the boundary of $E_m$.  \begin{figure}[h] \centering
 \includegraphics[width=0.8\linewidth]{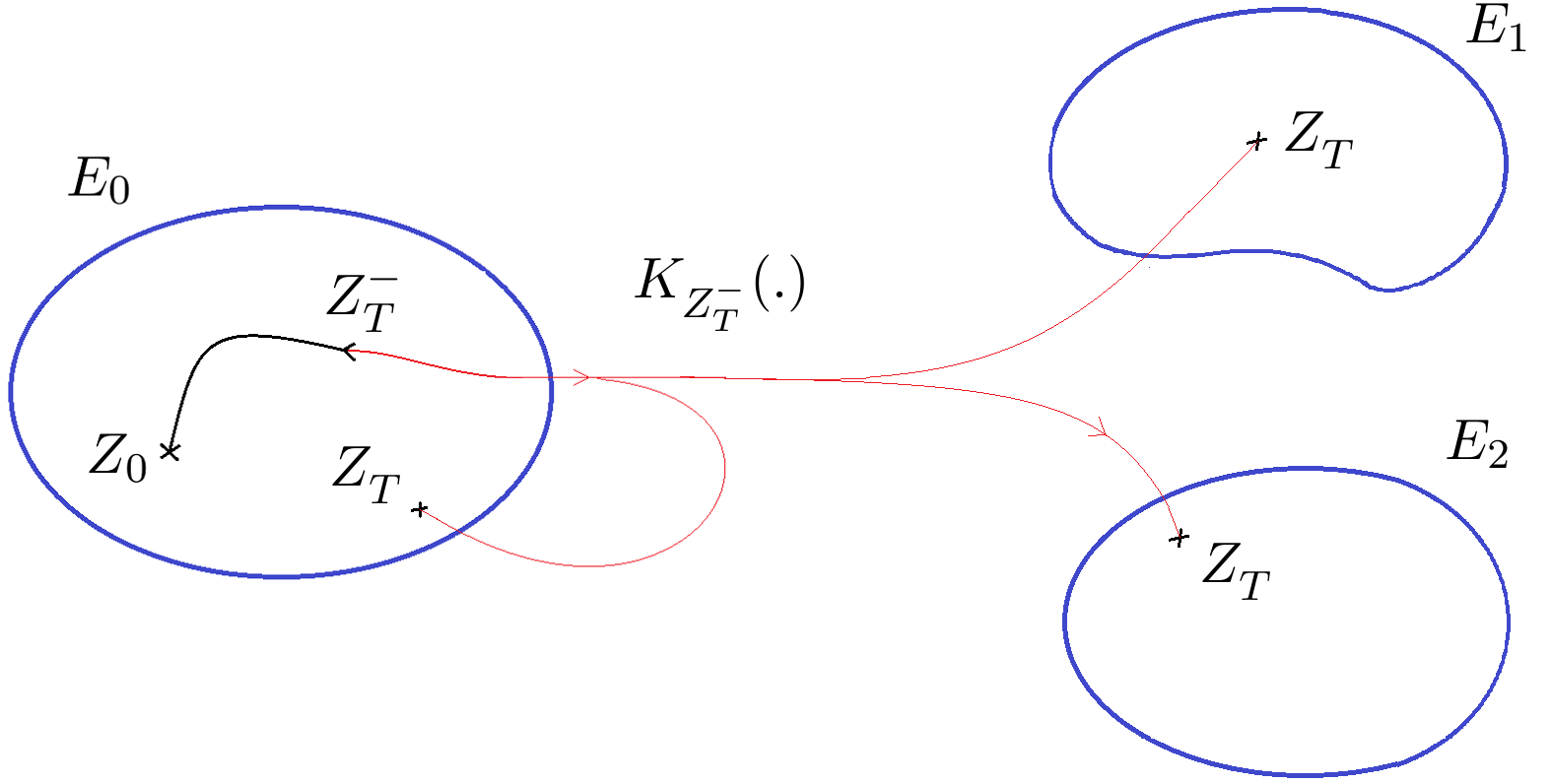}
 \caption{A spontaneous jump.\label{sautalea}} 
 \end{figure}The occurrence time of a spontaneous jump is modeled by a jump rate $\lambda(Z_{s+t})$. For instance, when we consider  that components fail or are repaired one at a time, this jump rate is the sum of the failure rates of the components in working condition and of the repair rates of the broken components.  In a more general case, each transition from a mode $m$ to a mode $m^+$ has its own rate $ \lambda_{m\to m^+}(X_{s+t})$, and we have \begin{equation}
 \lambda(Z_{s+t})=\sum_{m^+\in\mathbb M}\lambda_{M_s\to m^+}(X_{s+t})
\label{eq:lambda}
\end{equation}
 We  define the cumulative jump rate $\Lambda_{z}(t)$ by $\Lambda_{z}(t)=\int_0^t \lambda\big(\Phi_z(u)\big)du$, so that  $\forall t\in[0,T) ,\  \Lambda_{Z_s}(t) =\int_0^t \lambda\big(Z_{s+u}\big)du$. Eventually the cumulative distribution function (cdf) of $T$ (the time until the next jump starting form a state $Z_s=z$)  takes the form:\begin{equation}
\mathbb P (T\leq t|Z_s=z)=\left\{\begin{array}{cr}
1-\exp\left[-\Lambda_{z}(t)\right] & \mbox{ if }t<t^{*}_z\, ,\\
1& \mbox{ if }t\geq t^{*}_z\, .\\
\end{array}\right.
\label{Tlaw}
\end{equation}
Here   $t^{*}_z=\inf\{t>0,\Phi_z(t)\in \partial E_m\}$ is the time until the flow   hits the boundary starting from a state $z=(x,m)$. When there is no boundary i.e. $\{s>0,\Phi_z(s)\notin E_m\}=\emptyset$, we take the convention that $t^{*}_z= + \infty\ $.   With this definition, the law of $T$ has a continuous part associated to spontaneous jumps (spontaneous failures and repairs) and a discrete part associated to forced jumps (control mechanisms). A possible reference measure for $T$ knowing $Z_s=z$ is then
 \begin{align}
 \forall B\in \mathcal B(\mathbb R^+),\quad \mu_{z}(B) &= leb\left(B \cap( 0,t^*_z)\right)+\indic{t^*_z<\infty}\, \delta_{t^*_z}(B)\ , \label{measureMu}
 \end{align}
 where $leb(.)$ corresponds to the Lebesgue measure. \\
If a jump occurs at time $S$, then the distribution of the destination of the jump is expressed by a transition kernel $\mathcal K_{Z_S^{\text{-}}}$ where $Z_S^{\text{-}}$ is the departure state of the jump. If $ \xoverline{ E\,}$ is the closure of $E$, we have $Z_S^{\text{-}}\in\xoverline{ E\,}$. Let $Z_S^+\in E$ be the arrival state, and  $\mathscr B(E)$ be the Borelian $\sigma$-algebra on $E$. Then, the law of a jump from a departure state $z^{\text{-}}$ is  defined by:\begin{align}
 \forall B\in \mathscr B (E),\qquad \mathbb P \left(Z_T^+\in B|Z_T^{\text{-}}=z^{\text{-}}\right) &= \mathcal K_{z^{\text{-}}}(B).
 \end{align} 
For any departure state $  z^{\text{-}}\in\xoverline{ E\,}, $ we define  $x_a(z^{\text{-}},m^+)$   as the 
arrival for the vector of physical variables when the transition $m^{\text{-}} \to m^+$ is triggered from a state $z^{\text{-}}$. We   define by $z_a(z^{\text{-}},m^+)=\big(x_a(z^{\text{-}},m^+),m^+\big)$ the arrival state when the transition $m^{\text{-}} \to m^+$ is triggered from a state $z^{\text{-}}$.  Using this definition, for any departure state $  z^{\text{-}}\in\xoverline{ E\,}, $ we define  the   measure $ \nu_{z^{\text{-}}}$ on $\big(E,\mathscr B(E)\big)$ by:\begin{equation}
   \forall B\in \mathscr B(E),\qquad \nu_{z^{\text{-}}}(B)=\sum_{m^+\in \mathbb M}\delta_{ z_a(z^{\text{-}} ,m^+)}(B).
\end{equation}
The kernel $\mathcal K_{z^{\text{-}}}$ is absolutely continuous with respect to the measure  $ \nu_{z^{\text{-}}}$. Denoting by $K_{z^{\text{-}}}$ its density with respect to $ \nu_{z^{\text{-}}}$ we have:
 \begin{align}
\mathbb P \left(Z_T^+\in B|Z_T^{\text{-}}=z^{\text{-}}\right) &=\int_B K_{z^{\text{-}}}(z^+)\,d \nu_{z^{\text{-}}}(z^+) \  .\label{Kernel} 
\end{align}
 Note   that we only consider systems for which  $\nu_{z^{\text{-}}}$ is discrete, whatever the departure state $z^{\text{-}}\in E$ may be. This hypothesis of discreteness of $\nu_{z^{\text{-}}}$ is mandatory to apply the method presented in Section \ref{sec:IPS+}.
 
  If $\nu_z$ has a Dirac contribution at point $z$, then the kernel must satisfy $K_z(z)=0$ , so it is not possible to jump on the departure state.  In some applications the physical variables are continuous, so that $K_{z^{\text{-}}}(z)$ is zero whenever $ x^-\neq x $. In some cases, one might want the model to include renewable and aging components. Then the vector $X_t$ should include the time since the last renewal of such component, and so the vector $X_t$ can be discontinuous at the time of a renewal. In both  situations, when $z^{\text{-}}$ is not on  the boundary of $E_{m^{\text{-}} }$, the jump kernel has this form:
\begin{align}
  \forall z^{\text{-}} \in E, \quad  B\in \mathscr B(E),\qquad   \mathcal  K_{z^{\text{-}}}(B)&=\sum_{m^+\in \mathbb M}\frac{\lambda_{m^{\text{-}} \to m^+}(x^{\text{-}})}{ \lambda(z^{\text{-}})} \delta_{ z_a(z^{\text{-}} ,m^+)}(B),  \label{eq:K_noboundary} \\
  \mbox{and}\qquad K_{z^{\text{-}}}(z^+)&= \frac{\lambda_{m^{\text{-}} \to m^+}(x^{\text{-}})}{ \lambda(z^{\text{-}})}\indic{x^+=x_a(z^{\text{-}},m^+)} .
\end{align}  When $z^{\text{-}}\in\partial E_{m^{\text{-}} }$, a control mechanism is triggered : some components  are required to turn on, or to turn off, so that the system reaches a desired state $z_c$. The transition to this state is usually very likely, so when $z^{\text{-}}\in\partial E_{m^{\text{-}} }$, the jump kernel $K_{z^{\text{-}}}$ tends to concentrate all or a big portion of the probability mass on this state $z_c$. In industrial systems there can be components which have a small probability to fail when they are asked to turn on. This phenomenon is referred as a failure on demand. We denote by $\kappa_i(z^-)$ the probability that the  $i^{th}$ component fails on demand when the control is triggered. We denote by $fod(i,m^{\text{-}},m)$   the indicator being equal to one if the $i^{th}$ component fails on demand during the transition from $m^{\text{-}}$ to $m$, and to zero otherwise. We also define   $ask(z^{\text{-}},z^+)$ as the set gathering the indices of the   components supposed to turn on during the control mechanism triggered from  $z^{\text{-}}$ to $z^+$. Finally, when $z^{\text{-}}\in\partial E_{m^{\text{-}} }$, we have that:
\begin{align}
  K_{z^{\text{-}}}(z^+)=
  \indic{x^+=x_a(z^{\text{-}},m^+)}
  \prod_{i\in ask(z^{\text{-}},z^+)}  
  \left(\kappa_i(z^-)\right)^{fod (i,m^{\text{-}},m^+)} \left(1-\kappa_i(z^-)\right)^{(1-fod (i,m^{\text{-}},m^+)) }    \label{eq:K_boundary}.
\end{align}\\

To generate  $\mathbf Z_t=(Z_s)_{s\in[0,t]}$ a trajectory  of the states of the system, one can repeat the following steps: Starting with $s=0$,
\begin{enumerate}
    \item Given a starting state $Z_s=z$, generate $T$ the time until the next jump using \eqref{Tlaw},
    \item  Follow the flow $\Phi$ until $s+T$ using  \eqref{eq:flow}, and set the departure state of the next jump as being $Z_T^{\text{-}} =\Phi_z(T)$,
    \item generate $Z_{s+T}$ the arrival state of the next jump using $K_{Z_T^{\text{-}}}$
    \item repeat starting with $s=s+T$ until one gets a trajectory of size $t$.
\end{enumerate}   Defined in this way, the process $Z_t$ is Markovian \cite{davis1984}. \\

\subsection{Law of trajectory}

Let us denote by $n(\mathbf Z_{t})$ the number of jumps in the trajectory $\mathbf Z_{t}$, and by $S_k$   the time of the $k^{th}$ jump in the trajectory $\mathbf Z_{t}$  (with the convention that $S_0=0$ and $S_{n(\mathbf Z_{t})+1}=t$).   $T_k=S_{k+1}-S_k$ denotes the time between two consecutive jumps. We define the $\sigma$-algebra $\mathscr S_t$ as the $\sigma$-algebra generated by the sets in 
 $\underset{\ n\in \mathbb N^* }{\bigcup} \mathscr B\Big(\Big\{\big(z_{k},t_k\big)_{k\leq n}\in(E\times\mathbb R^{*}_\text{+})^{n},\, \overset{n}{\underset{i=0}{\sum}} t_i=t\Big\}\Big)$, where $\mathscr B(.)$ indicates the Borelians of a set. Letting $\Theta_{t} :  \mathbf Z_{t}\to \left((Z_k,T_k)\right)_{0\leq k\leq n(\mathbf Z_{t})}$ be the application giving the skeleton of a trajectory, we can define the law of trajectories as an image law through $\Theta_{t}$. We have that for  $B\in \mathscr S_t$: 
 \begin{align}
 \mathbb P_{z_o}\Big(\mathbf{Z_t} \in \Theta_t^{-1}( B) \Big) = &\int_{B} \ \prod_{k=0}^{n} \Big(\lambda_{z_{k} }(t_k) \Big)^{\mathbbm{1}_{ t_k< t^*_{z_{k}}}}\exp\Big[-\Lambda_{z_{k}}(t_k)\Big] \prod_{k=1}^{n }K_{z_{k}^-}(z_{k})\nonumber\\
&\quad\times d\delta_{t^*_{n}}(t_n)\ d\nu_{z_{n }^-}(z_n)\ d\mu_{z_{n-1}}(t_{n-1}) \ ...\ d\nu_{z_{ 1}^-}(z_1)\ d\mu_{z_{o}}(t_{0})\ , 
\label{eq:loitraj}
 \end{align}
 where $z_j^-=\Phi_{z_{j-1}}(t_{j-1})$, and $t^*_{n}=t-\sum_{i=0}^{n-1}t_i$.\\
In the rest of the paper we will denote  $\mathbf E_{t}$ the set of trajectories of size $t$ that  satisfy \eqref{eq:flow}, and by $\mathscr M(\mathbf E_{t})$  the set of bounded $ (\Theta_t^{-1}(\mathscr S_t), \mathbf E_{t})$-measurable functions.
 
\subsection{A hybrid reference measure for the trajectories of PDMP}
As they involve hybrid reference measures for the times between jumps, PDMPs are very degenerate processes. As a result, some of their realizations   have a strictly positive probability to be simulated, but some other do not.   The trajectories  that concentrate a part of the probability mass on their own, i.e. that   verify $\mathbb P\big(\mathbf Z_{t}=\mathbf z_t\big)>0$,  are called the preponderant trajectories. For example,  assuming that $\lambda$ is positive, if the trajectory $\mathbf z_t$ involves no jump, it is preponderant  because we have $\mathbb P\big(\mathbf Z_{t}=\mathbf z_t\big)=1-\exp\big(-\Lambda_{z_0}(t)\big)>0$.  Also, as we consider that the reference measures $\nu_{z^{\text{-}}}$ are discrete, there can be an other type of preponderant trajectories, which are the trajectories that only  jump  on the boundaries $\partial E_m$. Indeed, such trajectories $\mathbf z_t$ would verify \begin{align}
\mathbb P\big(\mathbf Z_{t}=\mathbf z_t\big)= \prod_{k=0}^{n}    \exp\Big[-\Lambda_{z_{s_k}}(t^*_{z_{s_k}})\Big] \prod_{k=1}^{n }K_{z_{s_k}^-}(z_{s_k})>0,
\end{align} where $s_k$ is the time of the $k^{th}$ jump in $\mathbf z_t$. Conversely,   some realizations can be considered as negligible as they verify  $\mathbb P\big(\mathbf Z_{t}=\mathbf z_t\big)=0$. These trajectories are the ones that involve a spontaneous jump, i.e. a jump starting from the interior of a set $E_m$.  \\
This can be better understood by looking at the equation \eqref{eq:loitraj} which shows that the measure $\zeta_{z_0,t}$ defined by :
 \begin{align}
 \forall B\in \mathscr S_t ,\quad \zeta_{z_0,t} (\Theta^{-1}( B))= & \underset{ \mbox{\hspace{-12ex} } (z_{_k},t_{_k})_{k\leq n}
 \in B}{\int\quad d\delta_{t^*_{n}}(t_n)\ d\nu_{z_{n }^-}(z_n)}\ d\mu_{z_{n-1}}(t_{n-1}) \ ...\ d\nu_{z_{ 1}^-}(z_1)\ d\mu_{z_{o} }(t_{0}) 
 \label{eq:zeta} 
\end{align}
where $z_j^-=\Phi_{z_{j-1}}(t_{j-1})$, and $t^*_{n}=t-\sum_{i=0}^{n-1}t_i$ is a reference measure for the law of trajectories. This reference measure  highlights the fact that the trajectories with no jump or with only jumps on boundaries can concentrate some probability mass on their own, because they refer to a Dirac contribution of the measure $\zeta_{z_0,t}$. Indeed  for such trajectories  the times between two consecutive jumps verify $t_k=t_{z_{s_k}}^*$ $\forall k < n$, and so, they always refer to the discrete part of the measures   $\mu_{z_{k-1}}$  $\forall k < n$, therefore such trajectories are related to the discrete part of $\zeta_{z_0,t}$.  Equation \eqref{eq:zeta} also shows that  the remaining  probability mass is  distributed continuously among the trajectories with at least one spontaneous jump. Indeed, in a negligible trajectory, if for example the $k+1^{th}$ jump is a spontaneous jump, then $t_k$  relates to the continuous part of the reference measure $\mu_{z_{k-1}}$.\\ 

\subsection{Concentrated PDMP}
\label{subseq:ConcentratedPDMP}
For reliability assessment of a highly reliable system  one often models the  system by a PDMP with low jump rates and concentrated jump kernels on the boundaries. Indeed, the components of the system are often reliable and their repair takes time, hence the low jump rates, and as failures on demand during a control mechanism are unlikely the jump kernels on boundaries are concentrated on one safe arrival state (i.e. the state aimed by the control mechanism). We call this kind of PDMP a concentrated PDMP, mainly because the law of one trajectory concentrates a big part of its probability mass. \\

This trajectory happens to be the   trajectory with  no failure and no repair. As jump rates are low   the probability of not having a spontaneous jump is close to one. For instance at a $k+1$-th jump this probability verifies  $$\mathbb P_{z_{s_{k}}} \left(T_k=t^*_{z_{s_{k}}}\right)=\exp\Big[-\Lambda_{z_{s_{k}}}(t^*_{z_{s_{k}}})\Big] \simeq 1.$$ So only jumps on boundaries are   likely, and when the process hits a boundary $\partial E_m$, the arrival state aimed by the control mechanism is very likely. Denoting by $z_{s_k}$ this state for a $k$-th jump we have :  $$K_{z_{s_k}^-}(z_{s_k})\simeq 1.$$ So if $\mathbf z_t$ is a  trajectory with  no failure and no repair of reasonable size we have:\begin{align}
\mathbb P\big(\mathbf Z_{t}=\mathbf z_t\big)= \prod_{k=0}^{n}    \exp\Big[-\Lambda_{z_{s_k}}(t^*_{z_{s_k}})\Big] \prod_{k=1}^{n }K_{z_{s_k}^-}(z_{s_k}) \simeq 1.
\end{align}  
 In this article we adapt the IPS method for concentrated PDMPs.  It is interesting to make the connection between our work and the modified particle filter developed in \cite{fearnhead2003line} which can be applied to the related context of discrete time processes with discrete spaces. 
 
\subsection{Reliability assessment}
Let $t_f$ be an observation time. We denote by $\mathbf E_{t_f}$ the set of trajectories of size $t_f$ that verify \eqref{eq:flow}, and $h\in \mathscr M(\mathbf E_{t_f})$. The methods presented in this article can be used for the estimation of any quantity $p_h$ defined by: $$p_h=\mathbb E[h(\mathbf Z_{t_f})].$$ 
  We are interested in estimating the probability, noted $p_{\mathscr D}$, that the system fails before the final observation time $t_f$ knowing it was initiated in a state $z_0$. Letting  $\mathscr D$ be the set of trajectories of length $t_f$ which pass through the critical region $D \subset E$, we have  $$ p_{\mathscr D}=\mathbb P_{z_0}\big(\mathbf Z_{t_f}\in \mathscr D\big)=\mathbb E_{z_0}\big[\indic{\mathscr D}( \mathbf Z_{t_f} )\big]. $$ 
 Although our application relates  to the case  where $h=\indic{\mathscr D}$, the IPS method will be presented with an arbitrary (bounded) function $h$.

\section{The IPS method}
\label{sec:IPS}
\subsection{A Feynman-Kac model}
For a measure $\eta$ and a bounded measurable function $f$ we note $\eta(f)=\int f\,d\eta$.  For a measure $\eta$ and a kernel $V$,  $\eta V$ denotes the measure such that  $\eta V(h)=\int \int h(y)V(dy|x)d\eta(x)$, and for a bounded measurable function $f$, $V(f)$ is  the function such that $V(f)(x)=\int h(y)V(dy|x)$.  

Consider a subdivision of the interval $[0,t_f]$ into $n$ sub-intervals of equal lengths,  noted $[\tau_k,\tau_{k+1})$, and such that $0=\tau_0<\tau_1<\dots<\tau_{n\,\text{-}1}<\tau_n=t_f$.  
Let $V_k $ be the  Markovian transition measure extending a trajectory of size $\tau_k$  into a trajectory of size $\tau_{k+1} $, such that \begin{equation}
    \forall \mathbf z_{\tau_k}\in \mathbf E_{\tau_{k}},\  B\in\boldsymbol{\mathcal E}_{\tau_{k+1}},\qquad \mathbb P\left(\mathbf Z_{\tau_{k+1}}\in B|\mathbf Z_{\tau_{k}}=\mathbf z_{\tau_{k}}\right) =V_k(B|\mathbf z_{\tau_k}).
\end{equation} 

For each $k< n$  we denote   $ G_k$ the  potential function  on $\mathbf E_{\tau_k}$, such that: \begin{align}\forall \mathbf z_k\in\mathbf E_k,\quad G_k(\mathbf z_k)&\geq 0 .\end{align}
 These potential functions are the main inputs of the method. The choice of the potential functions is important, because it will ultimately determine the variance of the estimator of $p_h$ provided by the IPS method. A good choice depends on the system and on the target function $h$.
The potential functions are used to define the target probability measures  $\tilde{\eta}_k$ for each $k\leq n$ , such that: \begin{align}
 \tilde{\eta}_k(d\mathbf z_{\tau_k}) &\propto \prod_{s=0}^k G_s(\mathbf z_{\tau_{s}})  \prod_{s=0}^{k-1}  V_{s }( d\mathbf z_{\tau_{s+1}}|\mathbf z_{\tau_{s}} ),  \label{eq:target}
\end{align} 
or equivalently 
\begin{align}
\forall B \in \boldsymbol{\mathcal E_{\tau_{k}}},\qquad  \tilde{\eta}_k(B) &=\frac{\mathbb E\left[ \indic{B}(\mathbf Z_{\tau_{k}})\prod_{s=0}^k G_s(\mathbf Z_{\tau_{s}})\right] }{\mathbb E\left[ \prod_{s=0}^k G_s(\mathbf Z_{\tau_{s}}) \right]}.  \label{eq:target1}
\end{align}
  Originally the IPS method comes from filtering methods. Filtering methods aim at estimating the target measures and, in these methods, the potential functions are chosen so that the $\eta_k$ match the target measures. But in our cases we have no interest in estimating the target measures, we only want to estimate $p_h$. So the potential functions can be chosen more freely. They are used to propose a probabilistic representation of $p_h$ in terms of a selection+mutation dynamics, which makes it possible to build an estimator of $p_h$ with a reduced variance.  

We define the propagated target measures $\eta_k$ such that $\eta_0=\tilde{\eta}_0$ and for $k\geq 0$, $\eta_{k+1}=\tilde{\eta}_kV_k$. We have  : \begin{align}
  \eta_{k+1}( d\mathbf z_{\tau_{k+1}}) &\propto  \prod_{s=0}^{k} G_s(\mathbf z_{\tau_{s}})  \prod_{s=0}^{k }  V_{s }( d\mathbf z_{\tau_{s+1}}|\mathbf Z_{\tau_{s}} )    \label{eq:targetProp},
\end{align}
or   equivalently
\begin{align}
\forall B \in \boldsymbol{\mathcal E_{\tau_{k+1}}},\qquad  \eta_{k+1}(B) &=\frac{\mathbb E\left[ \indic{B}(\mathbf Z_{\tau_{k+1}})\prod_{s=0}^k G_s(\mathbf Z_{\tau_{s}})\right] }{\mathbb E\left[ \prod_{s=0}^k G_s(\mathbf Z_{\tau_{s}}) \right]}.  \label{eq:targetProp1}
\end{align} 
For $k=0$ we consider that $\eta_0=\delta_0$, but the methods would still be valid if we had $\eta_0 \neq \delta_0$.
We define $Q_{k}$ such that for $f\in\mathscr M(\mathbf E_{\tau_{k+1}})$, 
$$
Q_k(f)(\mathbf Z_{\tau_k})=\int_{\mathbf E_{\tau_{k+1}}}f(\mathbf z_{\tau_{k+1}})\mathbf V_k(d\mathbf z_{\tau_{k+1}}|\mathbf Z_{\tau_{k}})G_k(\mathbf Z_{\tau_{k}})$$ and set  $Q_{k,n}=Q_kQ_{k+1} \dots Q_n$. Let $\Psi_k$ be the application that transforms a measure $\eta$ defined on $\mathbf E_{\tau_k}$ into a measure $\Psi_k(\eta)$ defined on $\mathbf E_{\tau_k}$ as follows:  \begin{equation}
    \Psi_k(\eta)(f)=\frac{\int G_k(\mathbf z)f(\mathbf z) d\eta(\mathbf z) }{\eta(G_k)}.
\end{equation}
We say that $\Psi_k(\eta)$ gives the selection of $\eta$ through the potential $G_k$. Notice that $\tilde{\eta}_k$ is the selection of $\eta_k$ as $\tilde{\eta}_k=\Psi_k(\eta_k)$. The target distributions can therefore be built according to the following pattern, in which a propagation step follows a selection step:
$$\eta_k\overset{\Psi_k}{\xrightarrow{\mbox{\hspace{1cm}}}} \tilde{\eta}_k\overset{.V_k}{\xrightarrow{\mbox{\hspace{1cm}}}} \eta_{k+1}.$$
We also define the associated unnormalized  measures $\tilde{\gamma}_k$ and $\gamma_{k+1}$, such that for $f\in\mathscr M(\mathbf E_{\tau_{k}})$:
\begin{equation}
  \tilde{\gamma}_{k}(f)   = \mathbb E\left[f(\mathbf Z_{\tau_{k}})\prod_{s=0}^{k} G_s(\mathbf Z_{\tau_{s}})    \right]  \quad \mbox{and}\quad \tilde{\eta}_k(f) =\frac{\tilde{\gamma}_{k}(f)}{\tilde{\gamma}_{k}(1)}, \label{eq:unnorm1}
\end{equation}  and for $f\in\mathscr M(\mathbf E_{\tau_{k+1}})$:
\begin{equation}
  \gamma_{k+1}(f) =\mathbb E\left[f(\mathbf Z_{\tau_{k+1}})\prod_{s=0}^{k} G_s(\mathbf Z_{\tau_{s}})    \right]   \quad \mbox{and} \quad \eta_{k+1}(f)=\frac{\gamma_{k+1}(f)}{\gamma_{k+1}(1)}. \label{eq:unnorm2}
\end{equation}
Denoting  $f_h(\mathbf Z_{\tau_{n}})=\frac{h( \mathbf Z_{\tau_{n}} )}{\prod_{s=0}^{n-1} G_s(\mathbf Z_{\tau_{s}}) }$, notice that we have:
\begin{equation}
    p_h=\gamma_n(f_h)= \eta_n(f_h) \overset{n-1}{\underset{{k=0}}{\prod}}  \eta_k\big(G_k \big )\label{eq:ph}.
\end{equation}

\subsection{The IPS algorithm and its estimators}

The IPS method provides an algorithm to  generate weighted samples which approximate the probability measures $\eta_{k}$ and $ \tilde{\eta}_k$  respectively for each step $k$.  For the sample approximating $ \eta_{k }$, we denote $\mathbf Z_{\tau_k}^j$  the $j^{th}$ trajectory and   $ W_{k}^j$ its weight.  Respectively, for the sample approximating $ \tilde{\eta}_k$, we denote $\mytilde{\mathbf Z}_{\tau_k}^j$   the $j^{th}$ trajectory and $\mytilde{W}_{k}^j$  its associated weight. For simplicity reasons, in this paper, we  consider that the samples all contain $ N$  trajectories, but it is possible to modify the sample size at each step, as illustrated in \cite{lee2018variance}. 
The  empirical approximations of $\eta_{k}$ and $ \tilde{\eta}_k$   are denoted by $\eta_{k}^N$ and $ \tilde{\eta}_k^N$  and are defined by:
\begin{equation}
    \tilde{\eta}_k^{N}  =  \sum_{i=1}^{N}  \,\mytilde W_k^i\, \delta_{  \mytilde{\mathbf Z}_{\tau_k}^i} \quad\mbox{ and }\quad
    \eta_k^{N}  =  \sum_{i=1}^{N}  \,W_k^i\, \delta_{  \mathbf Z_{\tau_k}^i}\,.\label{eq:estiEta}
\end{equation}
So for all $k\leq n$ and $f\in \mathscr M( \mathbf E_{\tau_k})$,
\begin{equation}
    \tilde{\eta}_k^{N}  (f) =  \sum_{i=1}^{N}  \,\mytilde{W}_k^i\, f \big( \mytilde{\mathbf Z}_{\tau_k}^i \big) \quad\mbox{and}\quad
    \eta_k^{N} (f) =  \sum_{i=1}^{N}  \,W_k^i\, f \big( \mathbf {Z}_{\tau_k}^i \big)\,.\label{eq:estiEta2}
\end{equation}
By plugging these estimations into equations \eqref{eq:unnorm1} and \eqref{eq:unnorm1}, we get estimations for the unnormalized distributions. Denoting by  $\tilde{\gamma}_k^{N}$ and   $\gamma_k^{N}$ these estimations, for all $k\leq n$ and $f\in \mathscr M( \mathbf E_{\tau_k})$, we  have:
\begin{equation}
    \tilde{\gamma}_k^{ N }  (f) = \tilde{\eta}_k^{N }(f)  \prod_{s=0}^{k-1} \eta_s^{N }(G_s)
    \quad \mbox{ and } \quad 
    \gamma_k^{  N } (f) =\eta_k^{N }(f)                   \prod_{s=0}^{k-1} \eta_s^{N}(G_s).
    \label{eq:estigamma}
\end{equation}
Plugging the estimations  $\eta_k^N$ into equation \eqref{eq:ph}, we get an estimator $\hat p_h$ of $p_h$ defined by:
\begin{equation}
\hat p_h = \eta_{n}^{N}(f_h)  \overset{n-1}{\underset{{k=0}}{\prod}}    \eta_k^{N}\big(G_k \big ) . \label{eq:hatp}
\end{equation}\\
\begin{figure}[h]\fbox{
    \begin{algorithm}[H]
   \textbf{Initialization :} $ k=0, \ \forall j =1..N, \ \mathbf Z_0^j\overset{i.i.d.}{\sim} \eta_0$ and  $ W_0^j=\frac 1 N $, and $\mytilde W_{0}^j=\frac{  G_0( \mathbf Z_{0}^j)}{\sum_{s }    G_0( \mathbf Z_{0}^s)}$  \\
 \While{ $k<n$}{ 
   \textbf{Selection:}\\
   Sample $ (\tilde N_k^j)_{j=1..N }\sim Mult\big(   N , (\mytilde  W_{k}^j)_{j=1..N }\big)$\\
   $\forall j:=1.. N  ,\ \mytilde W_{k}^j:= \frac 1 {  N }$  \\
   
   \textbf{Propagation :} \\
 \For{$j:=1.. N$}{ Sample $\mathbf{Z}_{\tau_{k+1}}^j$ from $ \mathbf V_{k+1}(. | \mytilde{\mathbf{Z}}_{\tau_{k}}^j) $\\
						set $W_{k+1}^j=\mytilde W_k^j$ }
 \For{$i:=1.. N$}{Set $\mytilde W_{k+1}^i=\frac{W_{k+1}^i G_{k+1}( \mathbf Z_{\tau_{k+1}}^i)}{\sum_{j }  W_{k+1}^j G_{k+1}( \mathbf Z_{\tau_{k+1}}^j)}$  }
    \eIf{$\forall j,\ \mytilde W_{k+1}^j=0$}{$\forall q>k, $ set $ \eta_q^N=\tilde \eta_q^N=0$ and Stop }{$k:=k+1$}
   	}\end{algorithm}}
\caption{IPS algorithm}
\label{fig:algo1}
\end{figure} 
The algorithm builds the samples sequentially, alternating between a selection step and a propagation step.  The $k^{th}$ selection step transforms the   sample $(\mathbf Z_k^j,W_k^j)_{j\leq N}$ into the sample $(\mytilde{\mathbf Z}_k^j,\mytilde{W}_k^j)_{j\leq N}$. This  transformation is done with a multinomial resampling scheme: the $\mytilde{\mathbf Z}_k^j$'s are drawn with replacement from the sample $ (\mathbf Z_k^j)_{j\leq N}$, each trajectory $ \mathbf Z_k^j$ having a probability $ \frac{W_{k}^j G_k( \mathbf Z_{\tau_{k}}^j)}{\sum_{i=1}^N  W_{k}^i G_k( \mathbf Z_{\tau_{k}}^i)}$ to be drawn each time. We denote by $A^j_k$ the ancestor index of the $j^{th}$ trajectory in the selected sample, such that $\mytilde{\mathbf Z}_k^j=\mathbf Z_{\tau_{k}}^{A^j_k}$. We let  $\tilde N_k^j=card\{i,A^i_k=j\}$ be the number of times the particle $ \mathbf Z_k^j$ is replicated in the sample $(\mytilde{\mathbf Z}_k^j,\mytilde{W}_k^j)_j$, so  $ N =\sum_{j=1}^{N } \tilde N_k^j$.  After this resampling  the weights $\mytilde W_k^j$ are set to $\frac 1 N$. The interest of this selection by resampling is that it discards low potential trajectories and replicates high potential trajectories. So the selected sample focuses on trajectories that will have a greater impact on the estimations of the next distributions  once extended. \\
If one specifies potential functions that are not positive, there can be a possibility that at a step $k$ we get $\forall j, G_k(\mathbf Z_{\tau_{k}}^j)=0 , $ and so  the probability for resampling cannot be defined. When this is the case, the algorithm  stops and we consider that $\forall s\geq k$ the measures $\tilde{\eta}_s^N$ and $ \eta_s^N$ are equal to the null measure.\\
Then the $k^{th}$ propagation step transforms the   sample $(\mytilde{\mathbf Z}_k^j,\mytilde{W}_k^j)_{j\leq N}$, into the sample \linebreak $(\mathbf Z_{k+1}^j,W_{k+1}^j)_{j\leq N}$. Each  trajectory $ \mathbf Z_{k+1}^j $  is obtained by extending the trajectory   $ \mytilde{\mathbf Z}_k^j $  on the interval $[\tau_k,\tau_{k+1})$  using the transition kernel $V_k$. The weights satisfy $W_{k+1}^j=\mytilde{W}_k^j ,\ \forall j$. Then the procedure is iterated until the step $n$. The full algorithm to build the samples is displayed in Figure \ref{fig:algo1}.\\

For the sake of simplicity, we will make the following assumption: 
$\exists\,\varepsilon_1,\varepsilon_2\in\mathbb R^+$ such that $\forall \,\mathbf Z_{\tau_k} \in \mathbf E_{\tau_k}$:\begin{equation}
      \varepsilon_1>G_k( \mathbf Z_{\tau_k})>\varepsilon_2>0 \label{eq:(G)}\tag{(G)},
 \end{equation} 

\begin{Theorem}
 When (G) is verified 
 the estimator \eqref{eq:hatp} is  unbiased and strongly consistent. \label{consitant}
\end{Theorem}
 The proof of theorem \ref{consitant} follows from Theorems 7.4.2 and 7.4.3 in \cite{del2004feynman}. 
\begin{Theorem} 
When (G) is verified: \begin{equation}\sqrt N (\hat p_h- p_h)\overset{d}{\underset{N\to\infty}{\longrightarrow}} \mathcal N\big(0,\sigma^2_{IPS,G}\big) ,\end{equation}
where
\begin{align}
\sigma^2_{IPS,G}&=\sum_{k=0}^{ n-1}  \gamma_k(1)^2   \eta_{k} \Big(\big[ Q_{k,n}(f_h)-\eta_{k}Q_{k,n}(f_h)
\big]^2\Big) \\ &=\sum_{k=0}^{ n-1}  \left\{\mathbb E_{z_0}  \bigg[ \prod_{i=0}^{k-1} G_i (\mathbf Z_{\tau_i})\bigg]   \mathbb E_{z_0}  \bigg[  \mathbb E[h(\mathbf Z_{\tau_n})|\mathbf Z_{\tau_{k}} ]^2\, \prod_{s=0}^{k-1} G_s^{-1 }(\textbf Z_{\tau_s})\bigg]-   p_h^2\right\}.
\label{eq:var}
\end{align}
\end{Theorem}
A proof  of this CLT can be found in \cite{del2004feynman} chapter 9 at the theorem 9.3.1 . 
For  the estimation of the variance $\sigma^2_{IPS,G}$ we refer the reader to \cite{lee2018variance}. 

\subsection{Classical improvements of the IPS method: the SMC method, and its alternative resampling}

We have seen that the resampling steps have the advantage of replicating high potential trajectories and discarding low potential trajectories. However the resampling steps also introduce some additional fluctuations  into the estimation (see \eqref{eq:var}). So we would like to trigger them only when it is judicious. Typically, not when the potentials of all the trajectories are similar, as in this case there is not point in discarding or replicating some trajectories over others. In order to avoid pointless resampling, one can trigger the selection step only when the weights are unbalanced.  This is done in the Sequential Monte Carlo (SMC) algorithm with adaptive resampling presented in Figure \ref{fig:algo2}.  In this algorithm, the heterogeneity of the weights is quantified using the effective sample size. At the $k^{th}$ step the effective sample size  is defined by:\begin{equation}
ESS_k=\dfrac{\left(\sum_{j=0}^{N}  W_k^j \,G_k(\mathbf Z_{\tau_{k}}^j)\right)^2}{\sum_{i=0}^{N} \left(W_k^i\, G_k(\mathbf Z_{\tau_{k}}^i)\right)^2}.
\end{equation}Its value is  between $1$ and $N$ and it measures the homogeneity in the candidate weights $\frac{W_{k}^i G_k( \mathbf Z_{\tau_{k}}^i)}{\sum_{j }  W_{k}^j G_k( \mathbf Z_{\tau_{k}}^j)}$: when $ESS_k = N$  the weights are perfectly balanced and are all equal to $\frac 1 {N}$, and conversely when $ESS_k = 1$ all the weights are null except one, which concentrates the totality of the mass.  Therefore, one considers the weights are too unbalanced when $ESS_k \leq eN$ where $e\in[0,1]$ is a tuning parameter.
\begin{figure}[h]\fbox{
    \begin{algorithm}[H]
   \textbf{Initialization :} $ k=0, \ \forall i :=1..N, \mathbf Z_0^i=(z_0)$ and  $ W_0^i=\frac 1 N $, and $\mytilde W_{0}^i=\frac{  G_0( \mathbf Z_{0}^i)}{\sum_{j }    G_0( \mathbf Z_{0}^j)}$  \\
 \While{ $k<n$}{
   \textbf{Selection:}\\
  \eIf{$ESS_k \leq eN$}{ Sample
   $(\tilde N_k^j)_{j=1..N}\sim Mult\big( N, (\mytilde  W_{k}^j)_{j=1..N}\big)$ and set $\forall i=1.. n ,\ \mytilde W_{k}^i:= \frac 1 {N}$  
   }{
 \For{$i:=1.. N$}{set $\mytilde{\mathbf Z}_{\tau_{k}}^i:=  \mathbf Z_{\tau_{k}}^i$ }
  }
   \textbf{Propagation :} \\
 \For{$i:=1.. N$}{   Sample $\mathbf{Z}_{\tau_{k+1}}^j$ from $ \mathbf V_{k+1}(. | \mytilde{\mathbf{Z}}_{\tau_{k}}^j) $\\
						set $W_{k+1}^i=\mytilde W_k^i$ }
\For{$i:=1.. N$}{Set $\mytilde W_{k+1}^i=\frac{W_{k+1}^i G_{k+1}( \mathbf Z_{\tau_{k+1}}^i)}{\sum_{j }  W_{k+1}^j G_{k+1}( \mathbf Z_{\tau_{k+1}}^j)}$  }
    \eIf{$\forall j,\ \mytilde W_{k+1}^j=0$}{$\forall q>k, $ set $ \eta_q^N=\tilde \eta_q^N=0$ and Stop }{$k:=k+1$}
 } 
\end{algorithm}}
\caption{SMC algorithm with adaptive resampling steps}
\label{fig:algo2}
\end{figure}  
\\

Note that in the presented algorithm one can use alternative strategies to select high potential trajectories. Here, the presented algorithms include  a standard multinomial resampling procedure, but one can also use residual resampling  or stratified resampling without altering the properties of the estimator. Empirical results suggest that these alternative resampling schemes  yield estimations with smaller variances~\cite{douc2005comparison,hol2006resampling}.  There are also recent theoretical results on the performance of different resampling schemes \cite{gerber2017negative}.  
 MCMC steps with invariant distribution $\tilde \eta_k$ can also be included  in the algorithm after the resampling step. Some adaptations of the algorithm for parallel implementations have also been studied in \cite{verge2015modele} for instance.\\ 

\section{Choice of the potential functions}
\label{sec:chooseG}
 
Note that the variance of $\hat p_h$ depend on the number of subdivisions and on the choice of the potential functions. We display here an important hint on how to select potential functions that yield a small variance.  Indeed, the theoretical expressions of the potential functions that minimize the asymptotic variance of the IPS estimator are known \cite{chraibi2018optimal}.
\begin{Theorem}
For $k\geq 1$, let $G^{*}_k $ be defined by:
\begin{align}
 G^{*}_k(\mathbf z_{\tau_k}) =
  \sqrt{\frac{\mathbb E\Big[ \  \mathbb E\big[ h(\mathbf Z_{\tau_n})\big|\mathbf Z_{\tau_{k+1}}\big]^2\big|\mathbf Z_{\tau_k} = \mathbf z_{\tau_k}\Big]\ \ \ }
{\mathbb E\Big[ \ \mathbb E\big[ \, h(\mathbf Z_{\tau_n}) \big|\mathbf Z_{\tau_k} \big]^2  \big|\mathbf Z_{\tau_{k-1}} = \mathbf z_{\tau_{k-1}}\Big]}}
 \label{eq:optiG1}
\end{align} 
if $\mathbb E\Big[  \mathbb E\big[   h(\mathbf Z_{\tau_n}) \big|\mathbf Z_{\tau_k} \big]^2  \big|\mathbf Z_{\tau_{k-1}}  = \mathbf z_{\tau_{k-1}}\Big]  \neq 0 $,
and $ G^{*}_k(\mathbf z_{\tau_k}) = 0 $ otherwise.
For $k=0$, we define
\begin{equation}
 \ G^{*}_0(\mathbf z_{\tau_0} )=  \sqrt{\, \mathbb E\Big[\, \mathbb E\big[ h(\mathbf Z_{\tau_n})\big|\mathbf Z_{\tau_1}\big]^2\big|\mathbf Z_{\tau_0} = \mathbf z_{\tau_0}\Big]}. \label{eq:optiG2}
\end{equation}  
The potential functions minimizing $\sigma^2_{IPS,G}$ are the ones that are proportional to the $ G^*_k $'s $\forall k\leq n$. The optimal variance of the IPS method with $n$ steps is then 
\begin{align}
&\sigma^2_{IPS,G^*} = 
 \mathbb E \left[  \,  \mathbb E\big[ h(\mathbf Z_{\tau_n})\big|\mathbf Z_{\tau_0}\big]^2 \right]-  p_h^2\nonumber\\
&+\sum_{k=1}^{ n } \left\{\mathbb E \left[\sqrt{\mathbb E\Big[ \  \mathbb E\big[ h(\mathbf Z_{\tau_n})\big|\mathbf Z_{\tau_k}\big]^2\big|\mathbf Z_{\tau_{k-1}}\Big]}\right]^2-  p_h^2\right\}.
\label{eq:varOpti}
\end{align}   
\end{Theorem}

For the reliability assessment case, the optimal target distributions have the form:
\begin{equation}
  \tilde \eta^{*}_k(d\mathbf z_{\tau_k}) \propto \sqrt{ \mathbb E\big[\mathbb E[ \indic{\mathscr D}(\mathbf Z_{\tau_n})|\mathbf Z_{\tau_{k+1}}]^2\big|\mathbf Z_{\tau_k}=\mathbf z_{\tau_k}\big]}\, \prod_{s=1}^k  V_s(d\mathbf z_{\tau_s}|\mathbf z_{\tau_{s\,\text{-}1}} ) ,
\end{equation} and the  potential functions would be defined by \eqref{eq:optiG1} and \eqref{eq:optiG2}, taking $h=\indic{\mathscr D}$. As we do not have  closed-form expressions of the functions $\mathbf z_{\tau_k}\to \sqrt{\mathbb E\big[\mathbb E[ \indic{\mathscr D}(\mathbf Z_{\tau_n})|\mathbf Z_{\tau_{s+1}}]^2|\big|\mathbf Z_{\tau_s}=\mathbf z_{\tau_k}\big]}$, we propose  to use instead a parametric approximation of these expectations based on our knowledge of the system and denoted by $U_\alpha(\mathbf Z_{\tau_s})$ so that we take \begin{align}
\tilde \eta_k(d\mathbf Z_{\tau_k}) &\propto U_\alpha(\mathbf Z_{\tau_k})\, \prod_{s=1}^k  V_s(d\mathbf Z_{\tau_s}|\mathbf Z_{\tau_{s\,\text{-}1}} )\\
\mbox{and } 
\forall s>0, \qquad G_s(\mathbf Z_{\tau_s})&=\frac{U_\alpha(\mathbf Z_{\tau_s})}{U_\alpha(\mathbf Z_{\tau_{s\,\text{-}1}})}\qquad\mbox{and}\qquad G_0(\mathbf Z_{\tau_0})=U_\alpha(\mathbf Z_{\tau_0}).
\end{align}
For a system including similar  components in parallel redundancy, we propose to set \linebreak $U_\alpha(\mathbf Z_{\tau_s})=1$ when $\mathbf Z_{\tau_s}$ has already reached the failure region once, and  to set $U_\alpha(\mathbf Z_{\tau_s})= \exp\big[-\alpha\, (b(Z_{\tau_s})+1)^2\big]\,L(\tau_s) $ otherwise, where $L$ is a positive function, and $b(Z)$ indicates the number of working components within a state $Z$. Here $\alpha$ is a parameter tuning the strength of the selection.

\section{The IPS+M method for concentrated PDMPs}\label{sec:IPS+}
\subsection{The issue with PDMP modeling a reliable system}\label{subsec:issue}
When it is used  on a reliable system and therefore on a concentrated PDMP (see Section \ref{subseq:ConcentratedPDMP}), the IPS method tends to loose in efficiency. This  efficiency loss can be attributed to the exploration steps. Remember that an  exploration step comes after a selection step: it builds a sample $ (\mathbf Z_{\tau_{k+1}}^j,  W_{k+1}^j )_{ j\leq  N}$ by extending the trajectories of a selected sample $(\mytilde{\mathbf Z}_{\tau_k}^j,\mytilde W_k^j )_{ j\leq \tilde N}$. This newly built sample $ (\mathbf Z_{\tau_{k+1}}^j,  W_{k+1}^j )_{ j\leq  N}$  fulfills two goals : 1) It contributes to the empirical approximation $\eta_{k+1}^{N}$  of $\eta_{k+1}$.   2)  It is used as a candidate sample for the next selection. But this second goal is often poorly achieved with a concentrated PDMP. Indeed, in order to get a good approximation $\tilde \eta_{k+1}^{N}$ of $\tilde \eta_{k+1}$, it is preferable that the candidate  sample to selection $(\mathbf Z_{\tau_{k+1}}^j,  W_{k+1}^j )_{ j\leq  N}$   contains  as many different trajectories as possible, along with high potential trajectories. Unfortunately, with this kind of PDMP, it is generally not the case: the candidate sample often  contains several replicates of the same trajectories, and no high potential trajectory.  Therefore each distribution $\tilde \eta_k $ is poorly represented, and so is each target distribution $\eta_{k+1}$, which eventually deteriorates the quality of the estimator $\hat p_{\mathscr D}$ or $\hat p_{h}$.\\

To understand why  the exploration steps   are not likely to generate many different trajectories with a concentrated PDMP, we have to  come back at the beginning of the  propagation step. At that point,  the  sample  $(\mytilde{\mathbf Z}_{\tau_k}^j,\mytilde W_k^j )_{1\leq j\leq N}$ is naturally clustered because of the previous selection step, each of the clusters containing  several replicates of the same trajectories. We can rewrite \eqref{eq:estiEta} in the following way \begin{align}
 \tilde{\eta}_k^N    &=  \sum_{i=1}^{N}    \,\mytilde W_k^{i}\, \delta_{  \mytilde{\mathbf Z}_{\tau_k}^{i}} = \frac{1}{N}\sum_{j=1}^{N}  \tilde N_k^j\delta_{  \mathbf Z_{\tau_k}^{j}} 
 \label{eq:etaClusterized}
\end{align}
where $\sum_{j=1}^{N} \tilde N_k^j=N$.  In practice many of the $\tilde N_k^j$ are null and only a few are positive and the $N$ resampled trajectories $(\mytilde{\mathbf Z}_{\tau_k}^j)_{j\leq N}$  are concentrated on a few trajectories.
Then, each of the $\tilde N_k^j$ trajectories of the $j$-th cluster is extended by using the same distribution $V_k(.| \mathbf Z_{\tau_k}^{j})$.  (For all index $i$  such that $A_{k}^{i}= j$ the trajectory $ \mytilde{\mathbf Z}_{\tau_k}^{i}$ is extended with the   kernel $V_k(.| \mathbf Z_{\tau_k}^{j})$). 
As the kernel $V_k(.| \mathbf Z_{\tau_k}^{j})$ corresponds to a concentrated PDMP, it is likely to extend all the trajectories of a cluster in the same manner. The trajectory $\mathbf a_{\tau_{k+1}}^{k,j}$ which extends  $\mathbf Z_{\tau_k}^{j}$ until   $ \tau_{k+1}$  without spontaneous jump  or   failure  concentrates the mass of the kernel $V_k(.| \mathbf Z_{\tau_k}^{j})$. Indeed, at this point we have : \begin{equation}V_k(\mathbf a_{\tau_{k+1}}^{k,j}| \mathbf Z_{\tau_k}^{j})=\mathbb P\big(
\mathbf Z_{\tau_{k+1}} =\mathbf a_{\tau_{k+1}}^{k,j} \big|
{\mathbf Z}_{\tau_{k}} =\mathbf Z_{\tau_k}^{j}\big)
\simeq 1.\end{equation}
Therefore each of the trajectories $\mytilde{\mathbf Z}_{\tau_{k}}^{i}= \mathbf Z _{\tau_{k}}^{j}$ in a cluster   tends to be extended in $\mathbf a_{\tau_{k+1}}^{k,j}$. Thus, the trajectories within a cluster are likely to stay clumped together during the propagation, and the propagated sample $ (\mathbf Z_{\tau_{k+1}}^j,  W_{k+1}^j )_{ j\leq  N}$ is very likely to be  clustered too. When the preponderant trajectories $\mathbf a_{\tau_{k+1}}^{k,j}$ have   low potential values, the sample is not likely to contain high potential trajectories. Consequently the selection step having no good candidates and too few candidates, it  tends to yield an inaccurate estimation  of the distributions $\tilde \eta_k$.\\
This situation is typical of reliability assessment. In that context, a  well constructed potential function is close to $G^*_k$ wherein $h=\indic{\mathscr D}$. So the potential of a trajectory $G_{k+1}(\mathbf Z_{\tau_{k+1}})$ should be high if its final state $Z_{\tau_{k+1}}$ is more degraded   than the state $Z_{\tau_k}$. This generally implies that  $\mathbf Z_{\tau_{k+1}} $ includes at least one component failure   between $\tau_k$ and $\tau_{k+1}$.  As the preponderant trajectories $ \mathbf a_{\tau_{k +1}}^{k,j}$  do not contain failure between $\tau_k$ and $\tau_{k+1}$ they generally are associated with low potential values.



The segment of $\mathbf a_{\tau_{k+1}}^{k,j}$ on  $(\tau_k,\tau_{k+1}]$ relates to a Dirac contribution of the measure $\zeta_{Z^j_{\tau_k},\tau_{k+1}-\tau_k}$. We can, therefore, decompose the expected propagation of the trajectory $\mathbf Z_{\tau_k}^{j}$  in this way:  \begin{align}
\delta_{ \mathbf Z_{\tau_k}^{j}}V_k(f) 
&=f\big(\mathbf a_{\tau_{k+1}}^{k,j}\big)V_k\big(\mathbf a_{\tau_{k+1}}^{k,j}|\mathbf Z_{\tau_{k}}^{j}\big)\ +\ 
\int_{\mathbf E_{\tau_{k+1}}\backslash\{\mathbf a_{\tau_{k+1}}^{k,j}\}} f\big( \mathbf z_{\tau_{k+1}} \big) V_k\big(d\mathbf z_{\tau_{k+1}}|\mathbf Z_{\tau_k}^{j}\big)  ,
\label{eq:clusterProp}
\end{align}
where $ f   \in\mathscr M (\mathbf E_{\tau_{k+1}})$. And the expected propagation of $\eta_k^N$ would be:\begin{align}
     \tilde \eta_{k}^{\tilde N}V_{k}(f)&=  \sum_{j=1}^{N} \frac{\tilde N_k^j}{ N} \ \delta_{\mathbf Z_k^j}V_{k}(f)\nonumber\\
     &= \sum_{j=1}^{N} \frac{\tilde N_k^j}{ N} f\big(\mathbf a_{\tau_{k+1}}^{k,j}\big)V_k\big(\mathbf a_{\tau_{k+1}}^{k,j}|\mathbf Z_{\tau_{k}}^{j}\big)\ +\ 
\frac{\tilde N_k^j}{ N}\int_{\mathbf E_{\tau_{k+1}}\backslash\{\mathbf a_{\tau_{k+1}}^{k,j}\}} f\big( \mathbf z_{\tau_{k+1}} \big) V_k\big(d\mathbf z_{\tau_{k+1}}|\mathbf Z_{\tau_k}^{j}\big)\nonumber\\
&\simeq \sum_{j=1}^{N} \frac{\tilde N_k^j}{ N} f\big(\mathbf a_{\tau_{k+1}}^{k,j}\big). \label{eq:clusteredPropIPS}
    \end{align}

 \subsection{Modify the propagation of clusters}
 
 In order to diversify the simulated trajectories, and to increase the precision of the estimation, we propose to modify the way we extend the selected trajectories.    Here we consider that  the size of the propagated sample  can differ from the size of the previous selected sample. We now denote $\tilde N_{k}$ the size of the $k^{th}$ selected sample, and $N_{k+1}$  the size of the $k^{th}$ propagated sample, with the convention $N_0=N$. We stressed out, in section \ref{subsec:issue}, that the   propagation step aims at providing an estimation of $\eta_{k+1}=\tilde \eta_{k}V_{k}$ using the selected sample $(\mytilde{\mathbf Z}_k^j ,\mytilde W_k^{j})_{j\leq \tilde N_k}$. In other words, the selection step aims at providing a  propagated weighted sample $(\mathbf Z_{\tau_{k+1}}^j,W_{k+1}^{j})_{j\leq N_{k+1}}$ to estimate the distribution $\tilde \eta_{k}^{\tilde N_k}V_{k}$ defined by: \begin{align}
     \tilde \eta_{k}^{\tilde N_k}V_{k}(f)&=\sum_{j=1}^{\tilde N_k} \mytilde W_k^{j}\ \delta_{ \mytilde{\mathbf Z}_k^j}V_{k}(f) =\sum_{j=1}^{N_k} \frac{\tilde N_k^j}{\tilde N_k} \ \delta_{\mathbf Z_k^j}V_{k}(f)\label{eq:clusteredProp} ,
    \end{align} where   $f\in\mathscr M(\mathbf E_{\tau_{k+1}})$.
   We denote by $\bar V_k$ the Markovian kernel from $\mathbf E_{\tau_{k}}$ to $\mathbf E_{\tau_{k+1}}$ such that, for any trajectory $ \mathbf Z_{\tau_{k}}^{j }$, $\bar V_k(.| \mathbf Z _{\tau_{k}}^{j } )$ is the conditioning of $V_k(.| \mathbf Z _{\tau_{k}}^{j } )$ to $\mathbf E_k\backslash \{\mathbf a_{\tau_{k+1}}^{k,j}\}$. $\bar V_k$'s density verifies:
 \begin{align}
   \bar V_k(\mathbf z_{\tau_{k+1}}| \mathbf Z _{\tau_{k}}^{j } )
   & =\frac{V_k(\mathbf z_{\tau_{k+1}}|\mathbf Z_{\tau_{k}}^{j} )}{1-V_k\big(\mathbf a_{\tau_{k+1}}^{k,j} \big| \mathbf Z _{\tau_{k}}^{j }   \big)} \indic{ \mathbf z_{\tau_{k+1}} \neq \mathbf a_{\tau_{k+1}}^{k,j}}.
 \end{align}
Using \eqref{eq:clusteredProp}  we can decompose $\tilde \eta_{k}^{\tilde N_k}V_{k}$ as follows:
 \begin{align}
     \tilde \eta_{k}^{\tilde N_k}V_{k}(f)
   &=\sum_{j=1}^{N_k} \frac{\tilde N_k^j}{\tilde N_k} \Bigg[V_k\big(\mathbf a_{\tau_{k+1}}^{k,j}|\mathbf Z_{\tau_{k}}^{j}\big)f\big(\mathbf a_{\tau_{k+1}}^{k,j}\big)  + \left(1-V_k\big(\mathbf a_{\tau_{k+1}}^{k,j}|\mathbf Z_{\tau_{k}}^{j}\big)\right) \delta_{\mathbf Z_k^j}\bar V_{k}(f)\Bigg].  \label{eq:probaA}
 \end{align}
 In the original IPS algorithm, the sample approximating $\tilde \eta_{k}^{\tilde N_k}V_{k}$ is built by directly extending  each trajectory in the selected sample.  When we extend the replicates of a cluster, in average a proportion $V_k(\mathbf a_{\tau_{k+1}}^{k,j}| \mathbf Z_{\tau_k}^{j})$  of the replicates are extended in $\mathbf a_{\tau_{k+1}}^{k,j}$. This proportion of trajectories extended in $\mathbf a_{\tau_{k+1}}^{k,j}$ then serves as an estimation of $V_k(\mathbf a_{\tau_{k+1}}^{k,j}| \mathbf Z_{\tau_k}^{j})$.   But it is not necessary to  misspend  all these replicates to estimate the probability of the preponderant trajectory. If we use   equation \eqref{eq:probaA},  we would need to generate the  trajectory $\mathbf a_{\tau_{k+1}}^{k,j}$ only once   to assess its contribution to the propagation of the  cluster.  Also, $\mathbf a_{\tau_{k+1}}^{k,j}$ is easy to get. To generate it, it generally suffices  to run the simulation process starting from the state $Z_{\tau_k}^j$ until time $\tau_{k+1}$, while setting   the jumps rates and the probability of failure on demand to zero.\\
 
 Therefore,  for each cluster, we propose to use an additional replicate to generate $\mathbf a_{\tau_{k+1}}^{k,j}$ and compute exactly its contribution. So for any $j\in\{1,\dots N_k\}$, we will extend the selected trajectory  $\mathbf Z_k^j$, $N_k^j$ times, where $N_k^j=\tilde N_k^j +\indic{\tilde N_k^j>0}$. We denote $j_i$ the index of the $i^{th}$ replicates of $\mathbf Z_k^j$, and consider the added replicate has  index $0$ such that for $i\in \{ 0, \dots , \tilde N^j_k\}$ we have $\mytilde{\mathbf Z}^{j_i}_{\tau_k}=\mathbf Z_k^j$. The additional replicate is deterministically extended to the preponderant trajectory, so we have $ \mathbf Z_{\tau_{k+1 }}^{j_{0}}=\mathbf a_{\tau_{k+1}}^{k,j}$, and we set its weight to $W_{k+1}^{j_0}=V_k(\mathbf a_{\tau_{k+1}}^{k,j}|\mathbf Z_{\tau_{k}}^{j}) \frac{\tilde N_k^j}{\tilde N_k}$, so that it carries all the mass associated to the preponderant trajectories of a cluster.
 Then we can use all the remaining $\tilde N_k^j$ trajectories in the cluster  to estimate the non preponderant part of the cluster's propagation (the $1^{st}$ term  in the right hand side of equation \eqref{eq:clusterProp}). For $i>0$, we  condition the extensions to avoid $\mathbf a_{\tau_{k+1}}^{k,j}$ generating the $\mathbf{Z\, }_{\tau_{k+1}}^{j_i}$ according to the kernel $\bar V_k(.| \mathbf{Z\,}_{\tau_{k}}^{j})$ and set  $W_{k+1}^{j_i}=\frac{1-V_k(\mathbf a_{\tau_{k+1}}^{k,j}|\mathbf Z_{\tau_{k}}^{j}) }{\tilde N_k^j}\frac{\tilde N_k^j}{\tilde N_k}$. \\
 Usually, the simulations of a restricted law  are carried out using a rejection algorithm, but in our case a rejection algorithm would perform poorly. The rate of rejection would be too high, as it would be  equal to $V_k(\mathbf a_{\tau_{k+1}}^{k,j}|\mathbf Z_{\tau_{k}}^{j}) $ which is typically close to 1. For PDMPs,  such  simulations, conditioned   to avoid a preponderant trajectory, can be efficiently carried out using the memorization method.  This method, introduced in \cite{labeau1996probabilistic}, shares similarities with the inverse method. It therefore benefits from not using any rejection, and so it is  well suited to our applications.  The memorization method  is presented in section \ref{sec:Memo}. 
 
The target distributions  $\tilde \eta_k $ are  still estimated with $ \tilde \eta_k^{\tilde N_k}$, using equation  \eqref{eq:etaClusterized}, but for $k=0 $ to $n-1$, $\eta_{k+1}$, the propagation of a target distribution, is now estimated by :
 \begin{align}
     \eta_{k+1}^{N_{k+1}}
     &=\sum_{i=1}^{N_{k+1}}W_{k+1}^{i}\delta_{\mathbf Z_{k+1}^i} = \sum_{\,j=1,\tilde N_k^j>0}^{N_{k }\,}  \sum_{i=0}^{\tilde N_k^j}W_{k+1}^{j_i}\delta_{\mathbf Z_{k+1}^{j_i}} \nonumber\\
     &= \sum_{\,j=1,\tilde N_k^j>0\,}^{N_{k }}\frac{\tilde N_k^j}{\tilde N_k} \Bigg[V_k\big(\mathbf a_{\tau_{k+1}}^{k,j}|\mathbf Z_{\tau_{k}}^{j}\big)\delta_{\mathbf a_{\tau_{k+1}}^{k,j}}  + \frac{\left(1-V_k\big(\mathbf a_{\tau_{k+1}}^{k,j}|\mathbf Z_{\tau_{k}}^{j}\big)\right) }{\tilde N_k^j}\sum_{i=1}^{\tilde N_k^j} \delta_{\mathbf Z_{k+1}^{j_i}} \Bigg]
 \end{align}

Let $\tilde{\mathbf N}_k=(N_0,\tilde N_0,N_1,\tilde N_1,\dots,\tilde N_k) $ and $ \mathbf N_k=(N_0,\tilde N_0,N_1,\tilde N_1,\dots,N_k) $
We now note $\tilde{\gamma}_k^{\tilde{\mathbf N}_k}$ and   $\gamma_k^{ \mathbf N_k}$ the  estimations of the unnormalized distributions, and  for all $k\leq n$ and $f\in \mathscr M( \mathbf E_{\tau_k})$, we  define them by:
\begin{equation}
    \tilde{\gamma}_k^{ \tilde{\mathbf N}_k }  (f) = \tilde{\eta}_k^{\tilde N_k }(f)  \prod_{s=0}^{k-1} \eta_s^{N_s }(G_s)
    \quad \mbox{ and } \quad 
    \gamma_k^{  \mathbf N_k } (f) =\eta_k^{ N_k }(f)                   \prod_{s=0}^{k-1} \eta_s^{N_s}(G_s).
    \label{eq:estigamma2}
\end{equation}
In the end, $p_h$ is  estimated using the equation : \begin{equation}
\hat p_h = \eta_{n}^{N_n}(f_h)  \overset{n-1}{\underset{{k=0}}{\prod}}    \eta_k^{N_k}\big(G_k \big ) .
\end{equation}
The full modified version of the algorithm   is presented in  Figure \ref{fig:algo3}. We call this modified version of the IPS algorithm the IPS+M algorithm.\\

Throughout the rest of the paper, the notation $\mathbb E_{_M}$ will indicate that the expectation is associated to the IPS+M method and $\mathbb E$ will still denote the expectation for the original IPS method.
   \begin{figure}[ht]
\fbox{\begin{algorithm}[H]
   \textbf{Initialization :} $ k=0, \ \forall j =1..N, \mathbf Z_0^j=(z_0)$ and  $ W_0^j=\frac 1 {N} $, and $\mytilde W_{0}^j=\frac{  G_0( \mathbf Z_{0}^j)}{\sum_{s }    G_0( \mathbf Z_{0}^s)}$  \\
 \While{ $k<n$}{\textbf{Selection:}\\
    $\tilde N_k=N$, and sample 
   $ ( \tilde N_k^j)_{j=1..N_k}\sim Mult\big( N, (  \mytilde  W_{k}^j)_{j=1..N_k}\big)$\\
   $\forall i=1..\tilde N_k ,\ \mytilde W_{k}^i:= \frac 1 {N}$ \\
    $\forall i=1..\tilde N_k,\ $ set $ N_{k}^i=  \tilde N_k^i+\indic{\tilde N_k^i>0}$ \\ 
   set $   N_{k} :=\sum_{i=1}^{N_k}    N_{k}^i$       
   \textbf{Propagation :} \\
 \For{$j:=1..  N_k$
  }{\If{$N_k^j>0$}{set $ \mathbf{Z\, }_{\tau_{k+1}}^{j_0}=\mathbf a_{\tau_{k+1}}^{k,j}$  
    and $W_{k+1}^{j_0}=  V_k\big(\mathbf a_{\tau_{k+1}}^{k,j}|\mathbf Z_{\tau_{k}}^{j}\big)\sum_{i=1}^{\tilde N_k^j}\mytilde{W}_{k }^{j_i}$\\
        \For{$j=1.. \tilde N_k^j$
        }{Sample $ \mathbf{Z\, }_{\tau_{k+1}}^{j_i}$ from $ \bar V_k(.| \mathbf{Z\,}_{\tau_{k}}^{j})$ 
        and set $ W_{k+1}^{j_i}=\frac{(1-V_k\big(\mathbf a_{\tau_{k+1}}^{k,j}|\mathbf Z_{\tau_{k}}^{j}\big))}{ \tilde N_k^j }\sum_{i=1}^{\tilde N_k^j}\mytilde{W}_{k }^{j_i}$
        }}
   }
  \For{$i:=1..N_{k}$
  }{ $\mytilde W_{k+1}^i=\frac{W_{k+1}^i G_{k+1}( \mathbf Z_{\tau_{k+1}}^i)}{\sum_{j }   W_{k+1}^j G_{k+1}( \mathbf Z_{\tau_{k+1}}^j)}$}
   	 \eIf{$\forall j,\ \mytilde W_{k+1}^j=0$}{$\forall q>k, $ set $ \eta_q^{N_q}=\tilde \eta_q^{\tilde N_q}=0$ and Stop }{$k:=k+1$}}
\end{algorithm}}
\caption{IPS+M algorithm}
\label{fig:algo3}
\end{figure}

\subsection{Convergence properties of the IPS+M estimators}
In this section we show that  the estimator  $\hat p_h$ obtained with the IPS+M method have the same basic properties as the IPS estimator. With the IPS+M method,  $\hat p_h$ converges almost surely to $p_h$, it is unbiased, and it satisfies a CLT. The proofs that we provide in this section follow the reasoning of the proofs in \cite{del2004feynman}. We present how to adjust the original proofs to take  into account that the extensions of the trajectories within a cluster are no longer identically distributed. Finally we show  that the asymptotic variance of the CLT  is   reduced with the IPS+M method.\\

\subsubsection{The martingale decomposition of the anticipated biases}
With the IPS+M method, we assumed that $\forall k, \,\tilde N_k=N$. For $p\leq 2n,$ we define $\mathcal F_p$ the filtration associated to  the sequence of the $p$ first random samples  built with the IPS+M algorithm: $\big(( {\mathbf Z}_{\tau_0}^j)_{j\leq N_0},$ $(\mytilde{\mathbf Z}_{\tau_0}^j)_{j\leq N}, ({\mathbf Z}_{\tau_1}^j)_{j\leq N_1}\,,  \dots,\, \big)$. So when $p$ is an even number such that $p=2k$, $\mathcal F_p$ is adapted to the vector $\big(( {\mathbf Z}_{\tau_0}^j)_{j\leq N_0},$ $(\mytilde{\mathbf Z}_{\tau_0}^j)_{j\leq N},\, $ $\dots,\, ({\mathbf Z}_{\tau_k}^j)_{j\leq N_k}\big)$. For an odd number $p=2k+1$, $\mathcal F_p$ is adapted to the vector $\big(( {\mathbf Z}_{\tau_0}^j)_{j\leq N_0},(\mytilde{\mathbf Z}_{\tau_0}^j)_{j\leq N},$ $  \dots,\,({\mathbf Z}_{\tau_k}^j)_{j\leq  N_k},(\mytilde{\mathbf Z}_{\tau_k}^j)_{j\leq N}\big)$.  For $f\in \mathscr M(\mathbf E_{\tau_n})$ we let $\Gamma^{\mathbf N}_{p,2n} (h) $ be define by
\begin{align}
    \Gamma^{\mathbf N}_{2k,2n}(f)&=\gamma_k^{\mathbf N_k}(Q_{k,n}(f))-\gamma_k (Q_{k,n}(f)) \nonumber \\
    &=\gamma_k^{\mathbf N_k}(Q_{k,n}(f))- \gamma_n(f)
\end{align}
 and
\begin{align}
    \Gamma^{\mathbf N}_{2k+1,2n}(f)&=\tilde \gamma_k^{\tilde{\mathbf  N}_k}(V_k Q_{k+1,n}(f))-\tilde\gamma_k (V_k Q_{k+1,n}(f)) \nonumber \\
    &=\tilde \gamma_k^{\tilde{\mathbf  N}_k}(V_k Q_{k+1,n}(f))- \gamma_n(f).
\end{align}
Using a telescopic argument we get
\begin{align}
    \Gamma^{\mathbf N}_{p,2n}(f)=&\sum_{k=0}^{ \lfloor \frac{p} 2 \rfloor   }
    \gamma_k^{\mathbf N_k}\left(Q_{k,n}(f)\right)-\tilde \gamma_{k-1}^{\tilde{\mathbf  N}_{k-1}}\left(V_{k-1}Q_{k,n}(f)\right)
    \nonumber \\
    +&\indic{p>0}\sum_{k=1}^{  \lfloor \frac{p +1}{2}\rfloor    } \tilde \gamma_{k-1}^{\tilde{\mathbf  N}_{k-1}}\left(V_{k-1}Q_{k,n}(f)\right)-  \gamma_{k-1}^{\mathbf N_{k-1}}\left(Q_{k-1,n}(f)\right) ,
    \label{eq:Gammatelscop}
\end{align}
with the convention for $k=0,\ \tilde \gamma_{ -1}^{\tilde{\mathbf  N}_{ -1}}(V_{-1 }Q_{0,n}(f))=\gamma_n(f) $.\\
Noticing that $\gamma_k^{\mathbf N_k}(1)=\tilde \gamma_{k-1}^{\tilde{\mathbf  N}_{k-1}}(1)=\gamma_{k-1}^{\mathbf N_{k-1}}(G_{k-1})$,  we can rewritte   \eqref{eq:Gammatelscop}  as 
\begin{align}
    \Gamma^{\mathbf N}_{p,2n}(f)=&\sum_{k=0}^{\frac{\lfloor 2p \rfloor }{2}}\gamma_k^{\mathbf N_k}(1) \left(\eta_k^{N_k}(Q_{k,n}(f))- \tilde \eta_{k-1}^{\tilde N_{k-1}}V_{k-1}(Q_{k,n}(f))\right)\nonumber \\
    +&\indic{p>0}\sum_{k=1}^{\frac{\lfloor 2p + 1\rfloor }{2} }\tilde \gamma_{k-1}^{\tilde{\mathbf  N}_{k-1}}(1)\left( \tilde \eta_{k-1}^{\tilde N_{k-1}}(V_{k-1}Q_{k ,n}(f))- \Psi_{k-1}(\eta_{k-1}^{N_{k-1}})(V_{k-1}Q_{k,n}(f))\right) ,
\end{align}
where for $k=0$, we use the convention   $ \gamma_0^{\mathbf N_0}(1)\tilde \eta^{\tilde N_{-1}}_{-1}(V_{-1}Q_{0,n}(f)) =\gamma_n(f)$. The benefit of this decomposition is that it distinguishes the errors associated to the propagation steps  and the errors associated to the selection steps.   For the propagation steps, using \eqref{eq:clusterProp} we easily get that for any $f\in\mathscr M(\mathbf E_{\tau_{k+1}})$: \begin{equation}
    \mathbb E_M\left[  \eta_{k}^{N_k}(f)\ \big|\mathcal F_{2k-1} \right]
    =\tilde{\eta}_{k-1}^{\tilde N_{k-1}} V_{k-1}(f)
    \label{eq:unbiasedProp}.
\end{equation}
For the selection steps, as  the  resampling schemes are the same ones as for the IPS algorithm, we still have  for any $f\in\mathscr M(\mathbf E_{\tau_{k}})$: \begin{equation}
 \mathbb E_M\left[  \tilde{\eta}_{k}^{\tilde N_k}(f) \big| \mathcal F_{2k}\right]
 = \Psi_k(\eta_k^N)(f).
\end{equation} 
Thus, each selection step  and  propagation step is  conditionally unbiased.
Note that $\gamma_k^{\mathbf N_k}(1) $ is $\mathcal F_{2k-1}$-measurable  and $\gamma_{k}^{\tilde{\mathbf  N}_{k}}(1)$ is $\mathcal F_{2k}$-measurable, so,  when the samples are generated with  the IPS+M algorithm, $(\Gamma^{\mathbf N}_{p,2n}(h))_{p\leq2n}  $ is a $\mathcal F_p$-martingale. Therefore,  $\hat p_h$ stays unbiased with the IPS+M method, because $$\mathbb E_M[\Gamma^{\mathbf N}_{2n,2n}(f_h)]=\mathbb E_M[\hat p_h-p_h]=0.$$
 
\subsubsection{Almost sure convergence}

Thanks to this  martingale decomposition, we can use the same arguments as in the proof in the chapter 7  in \cite{del2004feynman}.  The hypotheses of Theorems 7.4.2 and 7.4.3 \cite[page 239 and 241]{del2004feynman} are satisfied  with the IPS+M method too, which yields the following theorem: 
\begin{Theorem}
For any $h\in\mathscr M (\mathbf E_{\tau_n}) $, $\hat p_h$ converges almost surely to $p_h$, and,   for any $f\in\mathscr M (\mathbf E_{\tau_k})$, $\eta_k^{N_k}(f)$ converges to $\eta_k(f)$ almost surely, $\gamma_k^{N_k}(f)$ converges to $\gamma_k(f)$ almost surely.
\end{Theorem}

\subsubsection{A Central Limit Theorem}
\begin{Theorem}
If the potential functions verify \ref{eq:(G)} and the samples are generated with the IPS+M algorithm ,then we have the following convergence in distribution:
$$
\sqrt{N}(\hat p_h-p_h)  \underset{N\to\infty}{\longrightarrow} \mathcal N\big(0,\sigma^2_{M,G}\big),
$$
where 
\begin{align}
\sigma^2_{M,G}=\ &\eta_0\left( \Big[Q_{0,n}(f_h)-\eta_0Q_{0,n}(f_h)\Big]^2\right)\nonumber\\
+&\sum_{k=1}^{n } \gamma_k(1)^2\tilde \eta_{k\text{-}1}\left( \big(1-V_{k-1}(\mathbf a_{\tau_{k}}| \mathbf Z_{\tau_{k\text{-}1}}) \big)^2   \bar V_{k\text{-}1}\Big[Q_{k,n}(f_h)- \bar V_{k\text{-}1}Q_{k,n}(f_h) \Big]^2 \right)\nonumber\\
+&\sum_{k=1}^{n }\tilde \gamma_k(1)^2\tilde \eta_{k\text{-}1}\left( \Big[V_{k\text{-}1} Q_{k ,n}(f_h)-\tilde \eta_{k\text{-}1} V_{k\text{-}1} Q_{k ,n}(f_h) \Big]^2\right) .
\end{align}
\end{Theorem}
 
\begin{proof}
This proof is very similar to what is done in Chapter  9 of \cite{del2004feynman}. In order to prove that $\hat p_h$ satisfies a CLT,  we begin by proving that the errors associated to the selection and propagation steps are normally distributed using  Lindeberg's theorem.\\
For a sequence of function $(f_k)_{k\leq 2n}$  such that $f_{2k}$ and $f_{2k+1}$ are in $\mathscr M(\mathbf E_{\tau_k})$, we define the sum of errors until   the $p^{th}$ selection and propagation by:
\begin{align}
    M^{\mathbf N}_{p,2n}(f)=
    & \sum_{k=0}^{ \lfloor\frac{ p }{2}  \rfloor  }        \eta_k^{N_k}(f_{2k})                    - \tilde \eta_{k-1}^{\tilde N_{k-1}}V_{k-1}(f_{2k}) \nonumber \\
    &+\indic{p>0} \sum_{k=1}^{ \lfloor\frac{ p+1 }{2}  \rfloor  } \tilde \eta_{k-1}^{\tilde N_{k-1}}(  f_{2k-1})- \Psi_{k-1}(\eta_{k-1}^{N_{k-1}})(f_{2k-1}) .
\end{align}
For $j \in \{1,\dots N\}$ we let 
\begin{equation}
 U_{(2k+1)N+j}^{N}(f ) =\frac 1 {\sqrt {N}} \left(f_{2k+1}(\mytilde{\mathbf Z}_k^j) -  \Psi_{k}(\eta_k^{N_{k}})(f_{2k+1}) \right).
 \end{equation}
 For $k\geq0$, $j \in \{1,\dots  , N_{k}\}$  and $i \in \{0,\dots ,  N^j\}$, we consider that the indices $j_i$ are ordered in such way that $j_0> N$ and $j_i< N$ when $i>0$. With such indexing $\forall s \in \{1,\dots,   N\}$, $\exists j \in \{1,\dots,   N_{k}\}$ and $i \in \{1,\dots,   N^j\}$ such that $s=j_i$, and for such $s$ we let  
 \begin{equation}
 U_{2kN+s}^{N}(f ) =\frac{1-V_k\big(\mathbf a_{\tau_{k+1}}^{k,j}|\mathbf Z_{\tau_{k}}^{j}\big)} {\sqrt {N}} \left(f_{2(k+1)}(\mathbf Z_{k+1}^{j_i}) -   \bar V_{k}(f_{2(k+1)})(\mathbf Z_k^{j}) \right).
 \end{equation}
For  $j \in \{1,\dots,   N_{0}\}$, let
  \begin{equation}
 U_{ j}^{N}(f ) =\frac{1} {\sqrt {N}} \left(f_{0}(\mathbf Z_{0}^{j}) -  \eta_0(f_{0}) \right).
 \end{equation}
 Thus,
  \begin{equation}
\sqrt N   M^{\mathbf N}_{p,2n}(f)=\sum_{k=0}^{(p+1)N}U_{k}^{N}(f).
 \end{equation}
 Noting $\mathcal P_k^N$ a filtration adapted to the $k$ first trajectories generated in the IPS+M algorithm. Note that we have that $
\mathbb E\left[U_{k}^{N}(f) |\mathcal P_{k-1}^N\right] =0$,  and $  \mathbb E\left[U_{k}^{N}(f)^2 |\mathcal P_{k-1}^N\right]<\infty$, 
and $ |U_{k}^{N}(f)|<\frac{2}{\sqrt{N}}\underset{k\leq  n, \mathbf  Z_{\tau_k}\in \mathbf  E_{\tau_k} }{\sup }\{|f_{2k}(\mathbf  Z_{\tau_{k}})| \wedge|f_{2k+1}(\mathbf  Z_{\tau_{k}}) |  \} 
$, so the Lindeberg condition is clearly satisfied. Then, we have that \begin{align}
  \langle \sqrt{N} M^{\mathbf N}_{p,2n}(f)\rangle_p  &=\sum_{k=0}^{(p+1)N}\mathbb E\left[U_{k}^{N}(f)^2|\mathcal P_{k-1}^N\right]\nonumber\\
  &  =\eta_0^N\left( \Big[f_{0}-\eta_0^N(f_{0})\Big]^2\right)\nonumber\\
&\quad +\sum_{k=1}^{\frac{\lfloor 2p   \rfloor}{2}  }  \tilde \eta_{k\text{-}1}^N\left( \big(1-V_{k-1}(\mathbf a_{\tau_{k}}| \mathbf Z_{\tau_{k\text{-}1}}) \big)^2  \bar V_{k\text{-}1}\Big[f_{2k}- \bar V_{k\text{-}1}f_{2k} \Big]^2\Big)\right)\nonumber\\
&\quad +\sum_{k=1}^{\frac{\lfloor 2p +1\rfloor}{2}} \tilde \eta_{k\text{-}1}^N\left( \Big[  f_{2k-1}(\mathbf Z_{\tau_{k\text{-}1}})-\Psi_{k\text{-}1}(\eta_{k\text{-}1}^{N_{k\text{-}1}})  f_{2k-1} \Big]^2\right) .  \nonumber\\
\end{align}
As $\eta_{k }^{N_{k }}$ and $\tilde \eta_{k }^{N }$ converge almost surely to $\eta_{k } $ and $\tilde \eta_{k } $,    $ \langle  \sqrt{N}M^{\mathbf N}_{p,2n}(f)\rangle_n  $ converge in probability to  
\begin{align}
  \sigma_p^2(f)
  &  =\eta_0\left( \Big[f_0-\eta_0 (f_0)\Big]^2\right)\nonumber\\
&\quad +\sum_{k=1}^{\frac{\lfloor  p   \rfloor}{2}  }  \tilde \eta_{k\text{-}1}\left( \big(1-V_{k-1}(\mathbf a_{\tau_{k}}| \mathbf Z_{\tau_{k\text{-}1}}) \big)^2 \bar V_{k\text{-}1}\Big[f_{2k}- \bar V_{k\text{-}1}f_{2k } \Big]^2\Big)\right)\nonumber\\
&\quad +\sum_{k=1}^{\lfloor \frac{ p +1}{2}\rfloor} \tilde \eta_{k\text{-}1}\left( \Big[ f_{2k-1}(\mathbf Z_{\tau_{k\text{-}1}})-\tilde\eta_{k\text{-}1} f_{2k-1} \Big]^2\right)  . \nonumber\\
\end{align}
By application of the Lindeberg's theorem for triangular array (see for instance  Theorem 4 on page 543 in \cite{shiryaev1996probability}),  we get that $\sqrt{N}M^{\mathbf N}_{p,2n}(f)$ converges in law to a centered Gaussian of variance $\sigma_p^2(f)$. As a corollary, if for $  p\neq2k$ we take $f_p=0 $ and for $p=2k$ $f_{2k}=Q_{k,n}(f_h) $, we get that \begin{align*} \sqrt N \bigg( \eta_k^{N_k} Q_{k,n}(f_h)  &- \tilde \eta_{k-1}^{N}V_{k-1}Q_{k,n}(f_h)\bigg) \\
&\underset{N\to\infty}{\longrightarrow}  \tilde \eta_{k\text{-}1}\bigg( \big(1-V_{k-1}(\mathbf a_{\tau_{k}}| \mathbf Z_{\tau_{k\text{-}1}}) \big)^2 \bar V_{k\text{-}1}\Big[Q_{k,n}(f_h) -\ \bar V_{k\text{-}1}Q_{k,n}(f_h)\Big]^2\Big)\bigg)
\end{align*}
and if for $p\neq 2k-1$ we take $f_p=0$ and for $p=2k-1$ $f_{2k-1}=V_{k-1}Q_{k,n}(f_h) $, we get that 
 \begin{align*} \sqrt N  \bigg(\tilde \eta_{k-1}^{N}( V_{k-1}Q_{k,n}(f_h))&- \Psi_{k-1}(\eta_{k-1}^{N_{k-1}})(V_{k-1}Q_{k,n}(f_h))\bigg)\\
 &\underset{N\to\infty}{\longrightarrow}\tilde \eta_{k\text{-}1}\bigg( \Big[ V_{k-1}Q_{k,n}(f_h)-\tilde \eta_{k\text{-}1} V_{k-1}Q_{k,n}(f_h) \Big]^2\bigg) .
 \end{align*} 
 Following from Theorem 4, $\gamma_k^{N_k}(1)^2 \ $ and  $\tilde \gamma_k^{N}(1)^2  $ converges almost surely  to $\gamma_k (1)^2 $ and  $\tilde \gamma_k(1)^2$ ,   by an application of Slutsky's Lemma, we get that $\sqrt N  \Gamma^{\mathbf N}_{2N,2n}(f_h )$ converges in law to a centered Gaussian with variance
 \begin{align}
  \sigma_{M,G}^2 
  &  =\gamma_0(1)^2\eta_0\left( \Big[Q_{0,n}(f_h)-\eta_0Q_{0,n}(f_h)\Big]^2\right)\nonumber\\
&\quad +\sum_{k=1}^{n } \gamma_k(1)^2 \tilde \eta_{k\text{-}1}\left( \big(1-V_{k-1}(\mathbf a_{\tau_{k}}| \mathbf Z_{\tau_{k\text{-}1}}) \big)^2 \bar V_{k\text{-}1}\Big[Q_{k,n}(f_h)- \bar V_{k\text{-}1}Q_{k,n}(f_h) \Big]^2\Big)\right)\nonumber\\
&\quad +\sum_{k=1}^{n} \tilde \gamma_{k\text{-}1}(1)^2\tilde \eta_{k\text{-}1}\left( \Big[ V_{k\text{-}1} Q_{k ,n}(f_h)-\tilde \eta_{k\text{-}1} V_{k\text{-}1} Q_{k ,n}(f_h) \Big]^2\right)  .
\end{align}
\end{proof}
\subsubsection{Variance reduction}
\begin{Theorem}

The variance of the original IPS can be decomposed as follows:
\begin{align}
\sigma^2_{IPS,G}&= \sigma^2_{M,G} + \sum_{k=1}^{n } \gamma_k(1)^2\tilde \eta_{k\text{-}1}\left(v_k(\mathbf Z_{\tau_{k\text{-}1}})  \bar V_{k\text{-}1}\Big(\Big[ Q_{k,n}(f_h)(\mathbf a_{\tau_{k}})- Q_{k,n}(f_h)(\mathbf Z_{\tau_k }) \Big]^2\Big)\right) ,\label{eq:stroumph3}
\end{align}
where $v_k(\mathbf Z_{\tau_{k\text{-}1}})=V_{k\text{-}1}(\mathbf a_{\tau_{k}}| \mathbf Z_{\tau_{k\text{-}1}})\big(1-V_{k\text{-}1}(\mathbf a_{\tau_{k}}| \mathbf Z_{\tau_{k\text{-}1}}) \big) $.
Therefore we have $\sigma_{M,G}^2\leq \sigma_{IPS,G}^2 $.
\end{Theorem}

 \begin{proof}

\begin{align}
\sigma^2_{IPS,G}&=\sum_{k=0}^{ n}  \gamma_k(1)^2   \eta_{k} \Big(\big[ Q_{k,n}(f_h)-\eta_{k}Q_{k,n}(f_h)
\big]^2\Big) \\
&=\eta_0\left( \Big[Q_{0,n}((f_h)-\eta_0Q_{0,n}((f_h)\Big]^2\right)\nonumber\\
&+\sum_{k=1}^{n }  \gamma_k(1)^2  \tilde \eta_{k-1}V_{k-1}\Big(\big[ Q_{k,n}(f_h)-V_{k-1}Q_{k ,n}(f_h)+V_{k-1}Q_{k ,n}(f_h)-\eta_{k}Q_{k,n}(f_h)
\big]^2\Big)\nonumber \\
\textbf{}&=\eta_0\left( \Big[Q_{0,n}((f_h)-\eta_0Q_{0,n}((f_h)\Big]^2\right)\nonumber\\
&+\sum_{k=1}^{n } \gamma_k(1)^2 \tilde \eta_{k-1}V_{k-1}\left( \Big[Q_{k,n}(f_h)-  V_{k-1}Q_{k ,n}(f_h) \Big]^2 \right)\nonumber\\
&+\sum_{k=1}^{n }\tilde \gamma_{k-1}(1)^2\tilde \eta_{k-1}\left( \Big[V_{k-1} Q_{k ,n}(f_h)-\tilde \eta_{k-1} V_{k-1} Q_{k,n}(f_h) \Big]^2\right) \label{eq:stroumph1}
\end{align}
 Temporarily using the notation $V_{k\text{-}1}(\mathbf a_{\tau_{k}}| \mathbf Z_{\tau_{k\text{-}1}}) =p_k $, for any $f\in\mathscr M(\mathbf E_{\tau_k})$, we get
\begin{align}
  V_{k-1}&\left( \Big[f(\mathbf Z_{\tau_k})-  V_{k-1}f \Big]^2\right)\nonumber\\
&=    V_{k-1}\left( \Big[f(\mathbf Z_{\tau_k})-p_k f(\mathbf a_{\tau_k})- (1-p_k)\bar V_{k-1}(f) \Big]^2\right)\nonumber\\
&=    V_{k-1}\left( f(\mathbf Z_{\tau_k})^2
-2p_k f(\mathbf Z_{\tau_k})f(\mathbf a_{\tau_k})
-2 (1-p_k) f(\mathbf Z_{\tau_k}) \bar V_{k-1}f \right.\nonumber\\
&\qquad\quad \left. + \,p_k^2 f(\mathbf a_{\tau_k})^2
+2p_k f(\mathbf a_{\tau_k})\bar V_{k-1}f
+(1-p_k)^2\big(\bar V_{k-1}f\big)^2\right)\nonumber\\
&=  p_k f(\mathbf a_{\tau_k})^2
+(1-p_k)\bar V_{k-1}(f^2)\nonumber\\
&\quad -2p_k^2  f(\mathbf a_{\tau_k})^2
-2p_k(1-p_k)  f(\mathbf a_{\tau_k}) \bar V_{k-1}(f)\nonumber\\
&\quad -2p_k (1-p_k)f(\mathbf a_{\tau_k}) \bar V_{k-1}f -2   (1-p_k)f^2 \bar V_{k-1}(f)^2  \nonumber\\
& \quad  + \,p_k^2 f(\mathbf a_{\tau_k})^2
+2p_k f(\mathbf a_{\tau_k})\bar V_{k-1}f
+(1-p_k)^2 \bar V_{k-1}(f)^2 \nonumber\\
&=  p_k (1-p_k)\big[ f(\mathbf a_{\tau_k})^2
-2f(\mathbf a_{\tau_k})\bar V_{k-1}f+\bar V_{k-1}(f^2)\big]\nonumber\\
&\quad+(1-p_k)^2\big(\bar V_{k-1}(f^2)-\bar V_{k-1}(f)^2\big)\nonumber\\
&=  p_k (1-p_k)\bar V_{k-1}\Big(\big[f(\mathbf a_{\tau_k}) -
 f(\mathbf Z_{\tau_k})\big]^2\Big) +(1-p_k)^2\Big(\bar V_{k-1}(f^2)-\bar V_{k-1}(f)^2\Big).\end{align}
  In particular, for $f=Q_{k,n}(f_h)$ we get
 \begin{align}
      V_{k-1} &\left( \Big[Q_{k,n}(f_h)-  V_{k-1}Q_{k,n}(f_h)\Big]^2\right)\nonumber\\
     &=p_k (1-p_k)\bar V_{k-1}\Big(\big[Q_k,n(f_h)(\mathbf a_{\tau_k}) -
 Q_k,n(f_h)\big]^2\Big) +(1-p_k)^2\Big(\bar V_{k-1}(f^2)-\bar V_{k-1}(f)^2\Big) \label{eq:stroumph2}
 \end{align}
Plugging \eqref{eq:stroumph2} into the second line of \eqref{eq:stroumph1}, yields \eqref{eq:stroumph3}.
 \end{proof}

 \section{Efficient extension of the trajectories using the Memorization method}\label{sec:Memo} 

 This section presents the  memorization method that was first introduced in \cite{labeau1996a}. Remember that we considered that a trajectory $\mathbf a_t$  is  preponderant whenever $p_{\mathbf a_t}=\mathbb P(\mathbf Z_t=\mathbf a_t)>0$. Assuming we know  such a preponderant trajectory $\mathbf a_t$, the memorization method allows to generate a trajectory $\mathbf Z_t$  which differs from this preponderant trajectory $\mathbf a_t$. 
  \subsection{Advantage of Memorization over a rejection algorithm}
 The interest of the method, compared to a rejection algorithm, is that we generate a trajectory $\mathbf Z_t\neq \mathbf a_t$ in one shot, whereas a rejection algorithm may generate several times the preponderant trajectory $\mathbf a_t$ before generating a trajectory different from $\mathbf a_t$. This is especially interesting when the probability $p_{\mathbf a_t}=\mathbb P(\mathbf Z_t=\mathbf a_t)$ is close to 1, as, with a rejection algorithm, the average number  of tries to  get a trajectory  different from $\mathbf a_t$    would be   $\frac{1}{1 - p_{\mathbf a_t}}$ which is then  very high. Therefore with a rejection algorithm much computational effort would be wasted generating $\mathbf a_t$ over and over.
 \subsection{Remarks on improving the IPS+M}
 Note that the IPS+M, greatly unbalances the weights of the propagated samples. Consequently it is   useless to consider an algorithm which triggers a resampling step, when the value of on the effective sample size is below a threshold. Indeed, the weights are so unbalanced that the effective sample size would always be very low and would trigger a resampling at each step.

  \subsection{Principle of the memorization method}
  
  \subsubsection{Work with the differentiation time}
The key idea of the memorization method is to consider the stopping time $\tau$ defined such that:\begin{equation}\forall s<\tau ,\quad   Z_s=  a_s \quad  \mbox{and}\quad  Z_{\tau}\neq a_{\tau}.\end{equation} This time $\tau$ is the time at which the trajectory  $\mathbf Z_{t}$  differentiates itself from   $\mathbf a_{t}$. So,  to generate $\mathbf Z_{t}$ knowing $\tau\leq t$ is equivalent to generate $\mathbf Z_{t}$ knowing it   differs from $\mathbf a_t$. In order to simulate a trajectory $\mathbf Z_{t}$ avoiding  $\mathbf a_{t}$, one can follow these three steps:
\begin{enumerate}
    \item generate $\tau$ knowing $ \tau\leq t$, and set $\mathbf Z_{\tau^{-}}=\mathbf a_{\tau^{-}  }$,
    \item generate $Z_{\tau}$ knowing $Z_{\tau}\neq a_{\tau}$,
    \item generate the rest of the trajectory normally until $t$.
\end{enumerate}
These steps are not difficult to realize, except for the first one.
 
\subsubsection{Generate $\tau$ knowing $ \tau\leq t$}

To achieve this first step, the authors in \cite{labeau1996a} propose to generate $\tau$ knowing $ \tau\leq t$   by using a method equivalent to the inverse  transform sampling method.  We present hereafter the theoretical foundation for this method.
We denote by $F$ the cdf of $\tau$ knowing $ \tau\leq t$:
\begin{equation}F(v)=\mathbb P\big(\tau<v | \tau \leq t\big)  ,\end{equation}
and   we denote by $F^{-1}$ its generalized inverse defined by
\begin{equation}
 F^{-1}(x)=\underset{v>0}{\inf}\{v\mid  F(v)\geq x \}. 
\end{equation}
We also denote by $\tilde F$ the  function defined by
\begin{equation}
    \tilde F (v)=\mathbb P (\mathbf Z_{v^-}=\mathbf a_{v^-})=\prod_{k=0}^{n(\mathbf a_v)}\exp\Big[-\Lambda_{a_{s_k}}(t_k)\Big] \prod_{k=1}^{n(\mathbf a_v) }\left(K_{a_{s_k}^-}(a_{s_k})\right)^{\indic{t_k>0}}
\end{equation}
where $\Theta_v(\mathbf a_v)=\big((a_{s_k},t_k)\big)_{0\leq k\leq n(\mathbf a_v)}$.   Note that $\tilde F$ is discontinuous in each jump times $s_k$ where $K_{z_k^-}(z_k)\neq 1$, so the inverse of $\tilde F$ is not necessarily defined everywhere on $[p_{\mathbf a_t},1]$. For this reason we consider $\tilde F^{-1}$, the  generalized inverse of $\tilde F$ defined by\begin{align}
    \tilde F^{-1}(x)=\underset{v>0}{\sup}\{v\mid \tilde F(v)\leq x \}. 
\end{align}  $\tilde F^{-1}$  extends the inverse of $\tilde F$ constantly  where it is not defined, this extension being done from the left so that $\tilde F^{-1}$ is right continuous.

The inverse  transform sampling method consists  in generating $U \sim Unif(0,1)$ and   taking   $F^{-1}(U)$ as a realization of $\tau\,|\, \tau\leq t$ which is a truncated random variable.  The simulation of such random variables is also presented in \cite{devroye1986sample}.   Note that the expression of the cdf $F$ can be related to $\tilde F$ , indeed we have:
\begin{equation}
    \forall v< t,\quad  F(v)=\frac{\mathbb P(\tau<v)}{\mathbb P(\tau \leq t)}=\frac{1-\mathbb P(\tau \geq v)}{1 - \mathbb P(\tau > t)}=\frac{1-\mathbb P (\mathbf Z_{v^-}=\mathbf a_{v^-})}{1- \mathbb P (\mathbf Z_t=\mathbf a_t)} = \frac{1-\tilde F(v)}{1- p_{\mathbf a_t}} .\nonumber
\end{equation}
Consequently we have that
\begin{align}
F^{-1}(U)&=\underset{v>0}{\inf}\{v\mid  F(v)\geq U \}\nonumber\\
&=\underset{v>0}{\sup}\big\{v\mid \tilde F(v)\leq 1-U\big(1-p_{\mathbf a_t})\big)\big\},\nonumber\\
&= \tilde F \Big( 1-U\big(1-p_{\mathbf a_t}\big)\Big).
\end{align} Also, as $U$ has uniform distribution on $[0,1]$,   $\tilde U=1 -U\big(1-p_{\mathbf a_t})$ is a uniform on $[p_{\mathbf a_t},1]$. Therefore,
sampling with the  inverse transform method is equivalent to simulating $\tilde U\sim Unif(p_{\mathbf a_t},1)$  and taking  $\tilde F^{-1}(\tilde U )$ as a realization of $\tau\,|\, \tau\leq t$.

  Assuming we first generate  the trajectory $\mathbf a_{t}$ and generate $\tilde U $  according to a uniform distribution on $(p_{\mathbf a_t},1)$, we now show  how to evaluate $\tilde F^{-1}(\tilde U)$. We consider that during the generation of $\mathbf a_{t}$, we computed   and memorized  $\mathbb P (\mathbf Z_{s_k^-}=\mathbf a_{s_k^-})$ and $\mathbb P (\mathbf Z_{s_k}=\mathbf a_{s_k})$  for each jump in the trajectory, and also for $\mathbb P (\mathbf Z_{t}=\mathbf a_{t})$. Then we   distinguish two cases: either there exists  $k\leq n(\mathbf a_t)$ such that $\mathbb P (\mathbf Z_{s_{k}^-}=\mathbf a_{s_{k}^-})\geq \tilde U> \mathbb P (\mathbf Z_{s_{k}}=\mathbf a_{s_{k}})<$, either there exists  $k\leq n(\mathbf a_t) $ such that $\mathbb P (\mathbf Z_{s_{k} }=\mathbf a_{s_{k} })\geq \tilde U> \mathbb P (\mathbf Z_{s_{k}^1}=\mathbf a_{s_{k+1}^-}) $ where we take the convention that $s_{n(\mathbf a_t)+1}^-=t$. The first case is quite simple as by definition of $\tilde F^{-1}$ we get $\tilde F^{-1}(\tilde U)=s_k$. In  the second case,  $\tilde F $ being continuous and strictly decreasing on $[s_k,s_{k+1})$, it is inversible on this interval, and $\tilde F^{-1} $   corresponds to $F$'s inverse on  $(\tilde F(s_{k+1}^-),\tilde F(s_{k})]$. So $\tilde F^{-1}(\tilde U)\in[s_k,s_{k+1})$ and $\tilde F(\tilde F^{-1}(\tilde U)) =\tilde U$. Notice that
  \begin{equation}
      \forall v \in[s_k,s_{k+1}),\quad \tilde F(v)=\tilde F(s_{k})\times \exp\left[-\Lambda_{a_{s_k}}(v-s_k)\right].
  \end{equation}
  So in particular, for $v=\tilde F^{-1}(\tilde U)$, we have :
  \begin{equation}
       \tilde U=\tilde F(\tilde F^{-1}(\tilde U))=\tilde F(s_{k})\times \exp\left[-\Lambda_{a_{s_k}}(\tilde F^{-1}(\tilde U)-s_k)\right],
  \end{equation}
  or equivalently
    \begin{equation}
      \log \left( \frac{\tilde F(s_{k})} {\tilde U}\right)  =\int_0^{\tilde F^{-1}(\tilde U)-s_k}\lambda_{a_{s_k}}(u)du.
    \end{equation}
  To determine $\tilde F^{-1}(\tilde U)$ we look for the value $s$ such that the integral  $\int_0^{s}\lambda_{a_{s_k}}(u)du$ is equal to $\log \left( \dfrac{\tilde F(s_{k})} {\tilde U}\right)  $ by dichotomy, then we set $\tilde F^{-1}(\tilde U)=s_k+s$.\\

  To sum up the generation of a realization of $\tau|\tau\leq t$  we   proceed as follow: \begin{enumerate}
    \item Generate $\tilde U \sim Unif(p_{\mathbf a_t},1)$, and set $k=0$
    \item If $\mathbb P (\mathbf Z_{s_{k}}=\mathbf a_{s_{k}})\geq \tilde U> \mathbb P (\mathbf Z_{s_{k+1}^-}=\mathbf a_{s_{k+1}^-})$, then we find $s\in[0,s_{k+1}-s_k)$ such that $$\log \left( \frac{\tilde F(s_{k})} {\tilde U}\right)  =\int_0^{s}\lambda_{a_{s_k}}(u)du ,
    $$ and we set $\tau=s_k+s$.
    \item If $ \mathbb P (\mathbf Z_{s_{k+1}^-}=\mathbf a_{s_{k+1}^-})\geq  \tilde U>\mathbb P (\mathbf Z_{s_{k+1}}=\mathbf a_{s_{k+1}})$, then $\tau=s_{k+1} $
    \item If the condition above were not satisfied, set $k=k+1$, if $k\leq n(\mathbf a_s)$ repeat the steps 2 to 4 
\end{enumerate}

\subsection{The memorization in the IPS+M algorithm}
In the IPS+M algorithm, assuming we apply the memorization method on an interval $(\tau_{k\,\text{-}1},\tau_k]$ and for the $i^{th}$ cluster, we apply it knowing that $\mathbf Z_{\tau_{k\,\text{-}1}}=\mathbf a^{i}_{k\,\text{-}1} $. Note that this trajectory $\mathbf a^{i}_{k\,\text{-}1}$ is not necessarily preponderant but when we extend it into $\mathbf a^{i}_{k}$, the piece of trajectory $\mathbf a^{i}_{(\tau_{k\,\text{-}1},\tau_k]})$ is  preponderant because :
\begin{equation}
\mathbb P(\mathbf Z_{(\tau_{k\,\text{-}1},\tau_k]}=\mathbf a^{i}_{(\tau_{k\,\text{-}1},\tau_k]}\big| \mathbf Z_{\tau_{k\,\text{-}1}}=\mathbf a^{i}_{k\,\text{-}1})=\mathbb P(\mathbf Z_{\tau_k}=\mathbf a^{i}_{\tau_k} \big| \mathbf Z_{\tau_{k\,\text{-}1}}=\mathbf a^{i}_{k\,\text{-}1})>0 .
 \end{equation}
So, in the IPS we try to generate trajectories of the cluster that verify $\mathbf Z_{\tau_{k\,\text{-}1}}=\mathbf a^{i}_{k\,\text{-}1}$, but avoid the piece of trajectory  $\mathbf a^{i}_{(\tau_{k\,\text{-}1},\tau_k]}$.

\section{Numerical illustrations}
\label{sec:Res}
\renewcommand{\arraystretch}{1.27}
In order to confirm our results numerically, we have applied the IPS method and the IPS+M method to two two-components system.
\subsection{The heated-room system}
The first system is a room heated by two heaters in passive redundancy.   Heaters are programmed to maintain the temperature of the room  above negative values, turning on when the temperature drops below some positive threshold and turning off when the temperature crosses a high threshold. The second heater can activate only when the first one is failed. The system fails when the temperature falls below zero.

  $X_t$ represents the temperature of the room at time $t$.  $M_t$ represents the status of the heaters   at time $t$.
  Heaters can be on, off, or out-of-order, so $\mathbb M=\{ON,OFF,F\}^{2}$. The state of the system is $Z_t=(X_t,M_t)$.
  
  The differential equation rule the temperature can be derived from the physics. $x_e$ is the exterior temperature. $\beta_1$ is the rate of the heat transition with the exterior. $\beta_2$ is the heating power of each heater. The differential equation giving the evolution of the  the temperature of the room  has the following form: $$\frac{d\,X_t}{dt}=\beta_1 (x_e-X_t)+\beta_2 \mathbbm 1_{M^1_t\, or\, M^2_t=ON}\ .$$  
 
 The heaters are programmed to maintain the temperature within an interval\linebreak$(x_{min} , x_{max})$ where $x_e<0<x_{min}$.  We consider that the two heaters are in passive redundancy in the sense that: when $X\leq x_{min}$ the second heater activates only if the first one is failed. When a repair of a heater occurs, if $X\leq x_{min}$ and the other heater is failed, then the heater status is set to $ON$, else the heater status is set to $OFF$.
 To handle the programming of the heaters, we set \scalebox{0.9}{$\Omega_{m}=(-\infty,x_{max})$} when all the heaters are failed $m=(F,F)$ or when at least one is activated, otherwise we set \scalebox{0.9}{$\Omega_{m}=(x_{min},x_{max})$}. \\
 Due to the continuity of the temperature, the reference measure for the Kernel is $\forall B\in\mathscr B (E),$ $\nu_{(x,m)}(B)=\sum_{m^+\in\mathbb M\backslash\{m\}}\delta_{(x,m^+)}(B)$.
 On the top boundary in $x_{max}$, heaters turn off with probability 1. On the bottom boundary in $x_{min}$, when a heater is supposed to turn on, there is a probability $\gamma=0.01$ that the heater will fail on demand. So, for instance, if $z^-=\big(x_{min}, (OFF,OFF)\big)$, we have $K_{z^-}\big(x_{min}, (ON,OFF)\big)=1-\gamma$, \linebreak and $K_{z^-}\big(x_{min}, (F,ON)\big)=\gamma(1-\gamma)$, and $K_{z^-}\big(x_{min}, (F,F)\big)=\gamma^2$.\\
 Let $j$ be a transition from $m$ to $m^+$. For the spontaneous jumps that happen outside boundaries, if the transition $j$ corresponds to the failure of a heater, then: \linebreak$\lambda^j(x,m)=0.0021+0.00015\times x $ 
and, if the transition corresponds to a repair, then $\lambda^j(x,m)=0.2\quad\mbox{when}\ M^j=F$.  Here the system failure occurs when the temperature of the room falls below zero, so $D=\{(x,m)\in E, x<0\}$. A possible trajectory of the state of this system is plotted in figure \ref{fig:schemetraj}.
The probability of failure $p$ was estimated to $2.71\times 10^{-5}$  thanks to a massive Monte-Carlo of $10^7$ simulations. 
\begin{figure}[h]\centering
 \includegraphics[width= 0.7\linewidth]{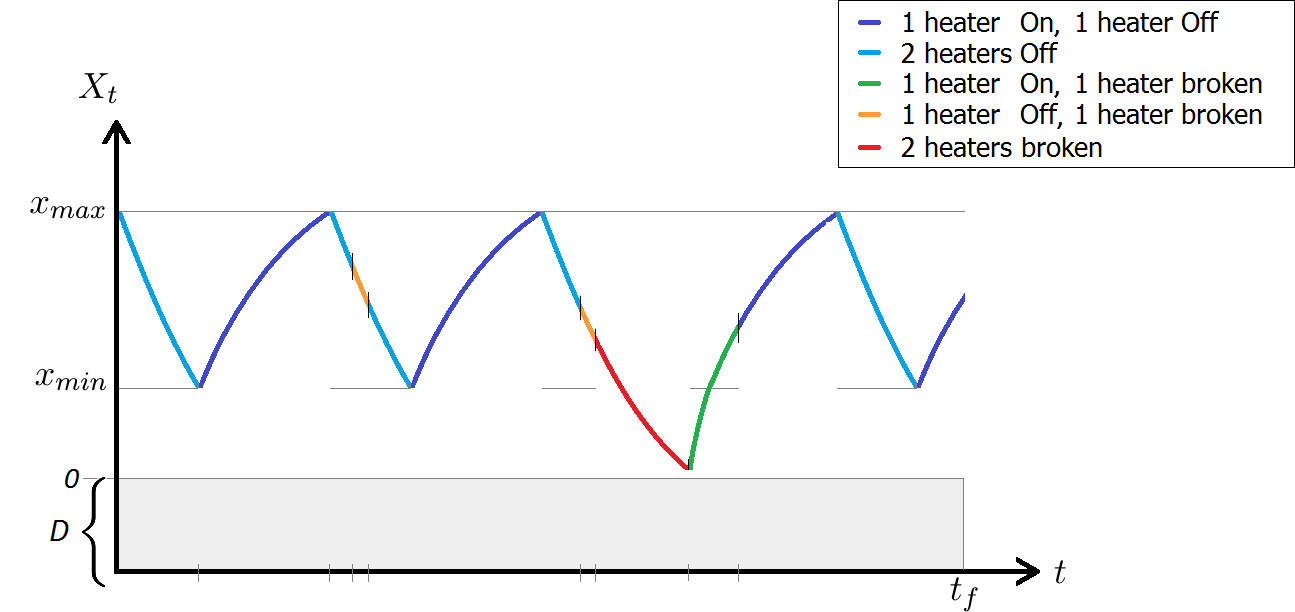} \caption{A possible trajectory of the heated-room system\label{fig:schemetraj}}{\footnotesize(the mode is represented with colors)}
 \end{figure}

\subsection{Results of the simulations for the heated-room system  }

 The results of the simulation study for the heated-room system are displayed in table \ref{table1}. Here we have used the potential functions proposed in section  \ref{sec:chooseG}. The value of $\alpha $ was set to 1.1. We have tried different values of $\alpha $ between $0.5$ and $1.5$ with a step $0.1$. The value of $\alpha=1.1$ was chosen among our trials as the one yielding the best variance reduction for the IPS method.  For the IPS, IPS+M, and MC methods the variances are estimated empirically: we run the methods 100 times and we take the empirical variances of the 100 estimates.  
The results highlight that the IPS method is ill-suited to PDMPs, as it yields a higher variance  than the MC method. Conversely, our IPS+M method performs well and has overcome the issue of the PDMP. Indeed, in the case $n=10$, it reduces the variance by a factor $2.7$ compared to the MC method, and by a factor $10$ compared to the IPS method. 
 
 The IPS+M is about 4 times slower than the IPS method, so, in terms of computational cost, the method is only $2.5$ more efficient than the IPS method on this test case. For a run of $N=10^5$ the IPS+M is  about $2.7$ time slower than Monte-carlo method. So  in terms of computational cost the IPS+M is slightly more efficient than the Monte-Carlo Method. 
 
\begin{table}[h]
\centering
\begin{tabular}{p{0.3cm}cc|c|c|c|p{0.6cm}@{}}
 \cline{4-6}
&  &  & MC  &  IPS  & IPS+M &\\
   \cline{2-6}
& \multicolumn{1}{|c|}{\multirow{2}{*}{ }}   & $\hat{p }$	&	  $\, 2.71\times10^{-5}$  & & &\\ 
&\multicolumn{1}{|c|}{}	 						& $\,\hat \sigma^2$ & $ 2.90\times10^{-10}$	&  &  &\\ 
\cline{2-6}
&\multicolumn{1}{|c|}{\multirow{2}{*}{$n=5$}}   	& $\hat{p }$ &		& $\, 2.86\times10^{-5}$ & $2.70\times10^{-5}$ &\\ 
&\multicolumn{1}{|c|}{}							& $\,\hat \sigma^2$ &	& $\, 1.78\times10^{-9}$ & $1.37\times10^{-10}$ &\\ 
\cline{2-6}
&\multicolumn{1}{|c|}{\multirow{2}{*}{$n=10$}}   & $\hat{p }$ &		& $\, 2.85\times10^{-5}$ & $ 2.64\times10^{-5}$&\\ 
&\multicolumn{1}{|c|}{}	 						& $\,\hat \sigma^2$ &	& $\, 1.08\times10^{-9}$ & $ 1.07\times10^{-10}$ &\\ 
\cline{2-6}
\end{tabular} 
\caption{ Empirical means and empirical variances on $100$ runs with $N=10^5$ for the MC, the IPS and the IPS+M methods\label{table1}}  
 \end{table}
 \subsection{Remark on the SMC with Memorization} 
 We have seen that it is possible to improve the IPS method to make it similar to the SMC method. We may, therefore, think that the IPS+M algorithm could be improved by adding adaptive optional re-sampling steps in order to get a SMC+M algorithm. In practice, however, it is not beneficial to add  these adaptive optional re-sampling steps.  Indeed we have noticed that, as we greatly modify the propagation process, the weights are greatly imbalanced and the effective-sample-size ends up being extremely small, which would trigger the re-sampling each time. Therefore adding adaptive optional re-sampling to the IPS+M has no effect, and in practice the IPS+M methods and the SMC+M methods are the same.

 \subsection{A dam system}
The second system models a dam subjected to an incoming water flow. The physical variable of interest is the water level in the dam denoted by $X_t$. The failure of the system occurs when the water level exceeds a security threshold $x_{lim}= 10$ before time $t_f=50$. The initial level is set to $X_0=0$. The water flow is characterized by the input debit $Q=10$. The dam has two evacuation valves with output debit $Q$. Each valve can be either  open, close or stuck closed. So $\mathbb M =\{Open, Closed, Stuck closed\}^2$. The valves are programmed in passive redundancy, so if the valves are in functioning order there is always one valve open and one valve closed. Though, the valve can get stuck closed and this happens at random times with exponential distribution with intensity $\lambda= 0.001$. The valves are repaired with a repair rate $\mu=0.1$. When both valves are stuck closed the  reservoir of the dam starts filling up according to the equation $\frac{dX_t}{dt}=Q/S $, where $S=10$ is the surface of the reservoir.

\subsection{Results of the simulations for the dam system  }

 The results of the simulation study for the dam system are displayed in table \ref{table2}. Here we have used the potential functions: 
 \begin{equation}
     \forall k<n,\quad G(\mathbf Z_{\tau_k})= \exp\left[\alpha_1(x_{lim}-X_{\tau_k})+\alpha_2(b(Z_{\tau_k})+1)^2\right].
 \end{equation}
The value of $\alpha_1 $ was set to $-0.9$ and the value of $\alpha_2$ was set to $-1$ (these are a priori guesses, we have not tried to use any optimization). For the  IPS and the IPS+M methods the variances are estimated empirically: we run the methods 50 times and we take the empirical variances of the 50 estimates. The results are presented in table \ref{table2}.
 \begin{table}[ht]
\centering
\begin{tabular}{p{0.3cm}cc|c|c|c|p{0.6cm}@{}}
 \cline{4-6}
&  &  & MC  &  IPS  & IPS+M &\\
   \cline{2-6}
& \multicolumn{1}{|c|}{\multirow{2}{*}{ }}   & $\hat{p }$	&	  $\,  1.12 \times10^{-4}$  & & &\\ 
&\multicolumn{1}{|c|}{}	 						& $\,\hat \sigma^2$ & $   1.12\times 10^{-9}$	&  &  &\\ 
\cline{2-6}
&\multicolumn{1}{|c|}{\multirow{2}{*}{$n=5$}}   	& $\hat{p }$ &		& $\, 1.75 \times10^{-4}$ & $ 1.12 \times10^{-4}$ &\\ 
&\multicolumn{1}{|c|}{}							& $\,\hat \sigma^2$ &	& $\, 3.08\times10^{-8}$ & $4.37\times10^{-9}$ &\\ 
\cline{2-6}

\end{tabular} 
\caption{ Empirical means and empirical variances on $50$ runs with $N=10^5$ for the MC, the IPS and the IPS+M methods\label{table2}}  
 \end{table}
 The results highlight that the IPS method is again ill-suited to PDMPs, as it yields a  variance 30 times larger than the MC method. Our IPS+M method performs better than the IPS method as the variance is reduced by a factor 7. Yet on this example the IPS+M method has not overcome the issue of the PDMP, as its variance is 3.4 times larger than the variance of the Monte-carlo estimator.   In terms of computational cost, on this example the IPS+M method was 3.6 times slower than the IPS, and 11.8 times slower then than the Monte-Carlo method. So the efficiency of the IPS+M is about 40 lower than the Monte-Carlo method. Clearly, the implementation of the IPS+M method requires a careful choice of the form of the potential functions and of their parameters.

 \section{Conclusion}
  This paper investigates the application of the IPS method to PDMPs.
 As the IPS method does not perform well when it is used on a concentrated PDMP, we introduce and analyze the IPS+M method, that is a modified version of the IPS that performs better with concentrated PDMP.  The IPS+M method is similar to the IPS but has different propagation steps. Its propagation steps focus on clusters of identical particles rather then on particles individually. For each cluster a memorization method is used to get an empirical approximation of the distribution of the propagated cluster, which allows to greatly improve the accuracy of the method. We have shown that the proposed algorithm yields a strongly consistent estimation, and that this estimation satisfies a TCL. We prove that  the asymptotic variance of the IPS+M estimator is always smaller than the asymptotic variance of the IPS estimator. Simulations also confirm these results, showing that the IPS+M can yield a variance reduction when the IPS cannot. In terms of computational cost, our implementations of the IPS+M method give approximately the same efficiency as the Monte-Carlo method in the examples considered in this paper, where the goal is to estimate a probability of the order of $10^{-5}$ for rather simple toy models. The numerical implementations certainly deserve more careful attention. We also believe that there are ways to improve the efficiency of the IPS+M method by finding a better class of potential functions. Another interesting improvement to the IPS+M method would be to propose an estimator of the variance.
We believe that it should be possible to adapt one of the estimators proposed in \cite{lee2018variance} for the IPS method in order to get an estimator of the variance for the IPS+M estimator. 
 
\newpage
\bibliographystyle{plain}
\bibliography{sample}
\end{document}